\documentclass{llncs}

\usepackage{tikz}
\usetikzlibrary{decorations.markings,fit,backgrounds}

\usepackage{amssymb, amsmath,microtype,comment,fullpage}

\title{Disconnected Cuts in Claw-free Graphs\thanks{This paper received support from the Leverhulme Trust (RPG-2016-258).}}

\author{Barnaby Martin\inst{1} \and Dani\"el~Paulusma\inst{1} \and Erik Jan van Leeuwen\inst{2}}
\institute{Durham University, Durham, United Kingdom\\
\texttt{\{barnaby.d.martin,daniel.paulusma\}@durham.ac.uk}
\and Utrecht University, Utrecht, The Netherlands,\\
\texttt{e.j.vanleeuwen@uu.nl}}

\usepackage{enumerate}
\usepackage{boxedminipage}

\newcommand{\NP}{{\sf NP}}

\newcounter{ctrclaim}[theorem]

\begin{document}
\maketitle

\begin{abstract}
A disconnected cut of a connected graph is a vertex cut that itself also induces a disconnected subgraph. The decision problem whether a graph has a disconnected cut is called {\sc Disconnected Cut}. This problem is closely related to several homomorphism and contraction problems, and fits in an extensive line of research on vertex cuts with additional properties. It is known that {\sc Disconnected Cut} is NP-hard on general graphs, while polynomial-time algorithms are known for several graph classes. However, the complexity of the problem on claw-free graphs remained an open question. Its connection to the complexity of the problem to 
contract a claw-free graph to the 4-vertex cycle $C_4$
led Ito et al.\ (TCS 2011) to explicitly ask to resolve this open question.

We prove that {\sc Disconnected Cut} is polynomial-time solvable on claw-free graphs, answering the question of Ito et al. 
The centerpiece of our result is a novel decomposition theorem for claw-free graphs of diameter~$2$, 
which we believe is of independent interest and expands the research line initiated by Chudnovsky and Seymour (JCTB 2007--2012) and Hermelin et al.\ (ICALP 2011). On our way to exploit this decomposition theorem, we characterize how disconnected cuts interact with certain cobipartite subgraphs, and prove two further novel algorithmic results, 
namely {\sc Disconnected Cut} is polynomial-time solvable on circular-arc graphs and line graphs.
\end{abstract}

\section{Introduction}\label{s-intro}

Graph connectivity is a crucial graph property studied in the context of network robustness. Well-studied notions of connectivity consider for example hamiltonicity, edge-disjoint spanning trees, edge cuts, vertex cuts, etc. In this paper, we study the notion of a \emph{disconnected cut}, which is a vertex set $U$ of a connected graph $G$ such that $G-U$ is disconnected and the subgraph $G[U]$ induced by $U$ is disconnected as well. Alternatively, we may say that $V(G)$ can be partitioned into sets $V_1,V_2,V_3,V_4$ such that no vertex of $V_1$ is adjacent to a vertex of~$V_3$ (that is, $V_1$ is anti-complete to $V_3$) and $V_2$ is anti-complete to $V_4$; then both $V_1 \cup V_3$ and $V_2 \cup V_4$ form a disconnected cut. See Figure~\ref{f-example} for an example. The {\sc Disconnected Cut} problem asks whether a given 
connected
 graph $G$ has a disconnected cut.

\begin{figure}[t]
 \centering
\includegraphics[scale=0.6]{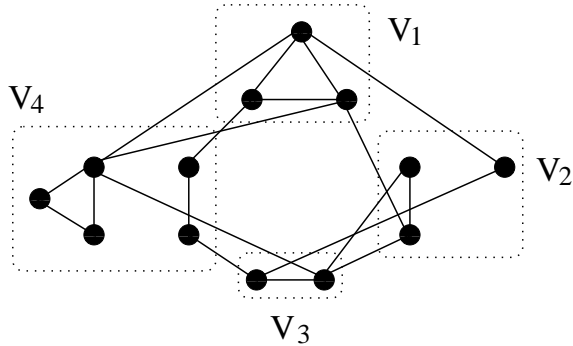}
\caption{Graph with disconnected cuts $V_1\cup V_3$ and $V_2\cup V_4$ (figure originally appeared in~\cite{IKPT11}).}
\label{f-example}
 \end{figure}

The {\sc Disconnected Cut} problem is intimately connected to at least five other problems studied in the literature. We give a brief overview here, and refer to the related work section for more details. The name {\sc Disconnected Cut} originates from Fleischner et al.~\cite{FMPS09}, who determined the complexity of partitioning the vertices of a graph into exactly $k$ {\it bicliques} 
(complete bipartite graphs with at least one edge),
except for the case $k=2$. For $k=2$, this problem is polynomially equivalent to {\sc Disconnected Cut} 
(by taking the complement of the input graph). The {\sc Disconnected Cut} problem can also be seen as an {\sc $H$-Partition} problem for appropriately defined 4-vertex graphs $H$. Dantas et al.~\cite{DFGK05} proved that {\sc $H$-Partition} is polynomial-time solvable for each 4-vertex graph $H$ except for the two cases equivalent to {\sc Disconnected Cut}. If the input graph has diameter~$2$, then {\sc Disconnected Cut} is equivalent to {\sc ${\cal C}_4$-Compaction}~\cite{FMPS09}, which 
asks for a homomorphism~$f$ from a graph~$G$ to the graph ${\cal C}_4$ (the 4-vertex-cycle with a self-loop in each vertex) such that for every $xy\in E(H)$ with $x\neq y$ there is an edge $uv\in E(G)$ with $f(u)=x$ and $f(v)=y$.
The diameter-$2$ case is also 
equivalent to testing if a graph can be modified to a biclique by a series of edge contractions~\cite{IKPT11}. The restriction to graphs of diameter~$2$ is natural, as {\sc Disconnected Cut} is trivial otherwise~\cite{FMPS09}.
Finally, {\sc Disconnected Cut} fits in the broad study of vertex cut problems with extra properties on the cut set,
such as $(k,\ell)$-cuts, $k$-cuts, (strict) $k$-clique cuts, stable cuts, matching cuts, etc.; see e.g.~\cite{KPST16} for an overview.

The above demonstrates that {\sc Disconnected Cut} is of central importance to understanding many different types of problems, ranging from cut problems to homomorphism and graph contractibility problems. Therefore, there has been broad interest to determine its computational complexity. Indeed, numerous papers~\cite{CDEFFK10,DFGK05,DMS12,Fi12,FMPS09,IKPTb11,IKPT11,TDF11} asked about its complexity on general graphs. \NP-completeness was proven independently in~\cite{MP15} and by Vikas, as announced in~\cite{Vi13}. The strong interest in {\sc Disconnected Cut} also led to a study on graph classes. We know polynomial-time algorithms for many classes, including graphs of bounded maximum degree, graphs not locally connected, graphs with a dominating edge (which include cobipartite graphs and $P_4$-free graphs)~\cite{FMPS09},
$2P_2$-free graphs, co-spiders, co-$P_4$-sparse graphs, co-circular arc graphs~\cite{CDEFFK10}, apex-minor-free graphs (which generalize planar graphs), chordal graphs~\cite{IKPT11},
$4P_1$-free graphs (graphs with independence number at most~4),  graphs of bounded treewidth, $(2P_1+P_2)$-free graphs (co-diamond-free graphs), $(C_5,\overline{P}_5)$-free graphs, co-planar graphs, co-$(q,q -4)$-graphs (for every fixed integer~$q$), and $(C_3+P_1)$-free graphs (which contains the class of triangle-free graphs)~\cite{DMS12}. The latter is the complement of the well-known class of claw-free graphs (graphs with no induced claw~$K_{1,3}$). 

Our interest in {\sc Disconnected Cut} on claw-free graphs is heightened by the close relation of this problem to $C_r$-{\sc Contractibility}, which is to decide if a graph~$G$ contains the 
$r$-vertex cycle $C_r$ as a contraction. This problem is \NP-complete if $r\geq 4$~\cite{BV87} and stays \NP-complete for 
claw-free graphs as long as $r\geq 6$~\cite{FKP13}. Given that the case $r \leq 3$ is polynomial-time solvable even for general graphs~\cite{BV87}, this leaves open on claw-free graphs the cases where $r\in \{4,5\}$. Ito et al.~\cite{IKPT11} showed that $C_4$-{\sc Contractibility} on claw-free graphs of diameter~$2$ is equivalent to {\sc Disconnected Cut}. As {\sc Disconnected Cut} is trivial if the input graph does not have diameter~$2$, this led Ito et al.~\cite{IKPT11} to explicitly ask the following:

\medskip
\noindent
{\it What is the computational complexity of {\sc Disconnected Cut} on claw-free graphs?}

\subsection{Our Contribution}

We answer the open question of Ito et al.~\cite{IKPT11} by giving a polynomial-time algorithm for {\sc Disconnected Cut} on claw-free graphs. This immediately implies that besides ${\cal C}_4$-{\sc Compaction}, also $C_4$-{\sc Contractibility} is polynomial-solvable on claw-free graphs of diameter~$2$, thus improving our understanding of these problems too. As claw-free graphs are not closed under edge contraction, the latter is certainly not expected beforehand.

Our result is grounded in a new graph-theoretic theorem that proves that claw-free graphs of diameter~$2$ belong to one of four basic graph classes after performing two types of elementary operations. We believe this novel structural theorem to be of independent interest. The theorem builds on one of the algorithmic decomposition theorems for claw-free graphs developed by Hermelin et al.~\cite{HMLW11-arxiv,HMLW11}, and relies on the pioneering works of Chudnovsky and Seymour (see~\cite{CS05}). Several other algorithmic decomposition theorems for claw-free graphs have been built on the ideas of Chudnovsky and Seymour, see e.g.~\cite{FOS14,Ki09}, which jointly have had a broad impact on our algorithmic understanding of claw-free graphs (see~\cite{HMLW11-arxiv} for an overview). Our structural theorem and resulting algorithm for {\sc Disconnected Cut} expand this line of research, and we hope it will prove similarly useful for future work.

The crux of the proof of our structural theorem is to exploit the extra structure offered by claw-free graphs of diameter~$2$ to show that the so-called strip-structures, which are central to the aforementioned decomposition theorems, only contain trivial strips. An important ingredient in the proof is to exclude not only 
twins (vertices $u,v$ for which $N[u] = N[v]$), but also vertices with nested neighbourhoods (vertices $u$ for which there exists a vertex $v$ such that $N(u) \setminus \{v\} \subseteq N(v) \setminus \{u\}$). Using this operation, one can simplify the decomposition theorem of Hermelin et al.~\cite{HMLW11-arxiv}, and thus we think this observation has an impact beyond this work. Indeed, our final decomposition for claw-free graphs of diameter~$2$ is much cleaner to state and easier to understand than the one for general claw-free graphs.

Using the structural theorem, {\sc Disconnected Cut} on claw-free graphs reduces to understanding its behavior under the elementary operations and on the basic graph classes. The crucial elementary operation is to remove certain cobipartite structures called W-joins. Intuitively, a W-join is a cobipartite induced subgraph such that each vertex of the rest of the graph is complete to one or two sides of the cobipartition of 
the W-join, or wholly anti-complete to the W-join. 
We develop the notion of unshatterable proper W-joins, which are essentially W-joins that cannot be broken into smaller W-joins, and exhibit how unshatterable proper W-joins interact with disconnected cuts. We then show that unshatterable proper W-joins can be removed from the graph by a simple operation. We complete our arguments by proving that all W-joins in the graph must be in fact be unshatterable proper W-joins, and that we can find unshatterable proper W-joins in polynomial time.

The main basic graph classes in the structural theorem 
 are line graphs and proper circular-arc graphs. Prior to our work, the complexity of {\sc Disconnected Cut} was unknown for these classes as well. We present a polynomial-time algorithm for line graphs and even for 
 general circular-arcs graphs 
(not only proper-circular arcs). 
 Both algorithms rely on  the existence of a small induced cycle passing through a disconnected cut in a highly structured matter. 
In addition, for line graphs, we prove that the pre-image of the line graph is $2P_2$-free, and thus has diameter at most~$3$. The hardest part of the proof is then to prove that if the pre-image has diameter exactly~$3$, then the line graph has in fact no disconnected cut.

\subsection{Related Work}

As mentioned, the name {\sc Disconnected Cut} stems from Fleischner et al.~\cite{FMPS09}, who studied how to partition the vertices of a graph into exactly $k$ bicliques, where {\sc Disconnected Cut} is equivalent to the case $k=2$. However, {\sc Disconnected Cut} originates from $H$-partitions, introduced in~\cite{DFGK05}. A model graph~$H$ on vertices $h_1,\ldots,h_k$ has solid and dotted edges. An {\it $H$-partition} of a graph $G$ is a
partition of $V(G)$ into $|V(H)|$ {\it nonempty} sets $V_1,\ldots,V_k$ such that for every pair of vertices $u\in V_i$ and $v\in V_j$: if $h_ih_j$ is a solid edge of $E(H)$, then $uv\in E(G)$; and if $h_ih_j$ is a dotted edge of $E(H)$, then $uv\notin E(G)$ (if $h_ih_j\notin E(H)$, then $uv\in E(G)$ or $uv\notin E(G)$ are both allowed). The corresponding decision problem is called $H$-{\sc Partition}.
 Dantas et al.~\cite{DFGK05} proved $H$-{\sc Partition} is polynomial-time solvable for every 4-vertex model graph~$H$ 
 except $H=2K_2$, which has solid edges $h_1h_3$, $h_2h_4$ and no dotted edges, and $H=2S_2$, which has dotted edges $h_1h_3, h_2h_4$ and no solid edges.
As a graph has a disconnected cut if and only if it has a $2S_2$-partition if and only if its complement has a $2K_2$-partition,
these two cases are polynomial-time equivalent to {\sc Disconnected Cut}. 
Hence, we now know that, as a matter of exception, $H$-{\sc Partition} is \NP-complete if $H\in \{2K_2,2S_2\}$~\cite{MP15}.

We can encode a model graph~$H$ as a matrix~$M$ in which every entry is either 0 (dotted edge), 1 (solid edge) or $*$ (no restriction). If we allow sets $V_i$ in a solution for $H$-{\sc Partition} to be empty, then we obtain the $M$-{\sc Partition} problem, introduced by Feder et al.~\cite{FHKM03}.
This well-known problem generalizes many classical problems involving vertex cuts and partitions, including {\sc $k$-Colouring} and {\sc $H$-Colouring}; 
see also~\cite{He14}.
An even more general variant is to give every vertex $u$ a list $L(u)\subseteq \{1,\ldots,k\}$ and to search for a solution, in which each vertex~$u$ may only belong to a set $V_i$ with $i\in L(u)$. This yields the {\sc List $M$-Partition} problem, which includes well-known cases, such as the {\sc Stubborn} problem, 
which turned out to be polynomial-time solvable~\cite{CPPW12}, in contrast to {\sc Disconnected Cut}. 
A homomorphism $f$ from $G$ to $H$ is a \emph{retraction} if $G$ contains $H$ as an induced subgraph and $f(u) = u$ for every $u\in V (H)$. The corresponding decision version is called {\sc $H$-Retraction}. Let ${\cal C}_4$ be the 4-cycle with a self-loop in each vertex. Then ${\cal C}_4$-{\sc Retraction} is a special case of {\sc List $2S_2$-Partition} where the input graph contains a cycle on four specified vertices $v_1,\ldots,v_4$ with $L(v_i)=\{i\}$ for $i=1,\ldots,4$ and $L(v)=\{1,2,3,4\}$ for $v\notin \{v_1,\ldots,v_4\}$. 
This problem is a generalization of {\sc Disconnected Cut}.
Feder and Hell~\cite{FH98} proved that {\sc ${\cal C}_4$-Retraction} is \NP-complete. Hence, {\sc List $2S_2$-Partition} and {\sc List $2K_2$-Partition} are \NP-complete. Note that this result is also implied by the \NP-completeness of $2K_2$-{\sc Partition}~\cite{MP15}.

Vikas~\cite{Vi02} solved an open problem of Winkler (see~\cite{FHKM03,Vi02}) by proving \NP-completeness of ${\cal C}_4$-{\sc Compaction}, the variant of the {\sc $2S_2$-Partition} problem with the extra constraint that there must be at least one edge $u_iu_j$ with $u_i\in V_i$ and $u_{i+1}\in V_{i+1}$ for $i=1,\dots,4$ (where $V_5=V_1$).  Generally, a homomorphism~$f$  from a graph~$G$ to a graph~$H$ is 
a {\it compaction} 
if $f$ is edge-surjective, i.e., for every $xy\in E(H)$ with $x\neq y$ there is an edge $uv\in E(G)$ with $f(u)=x$ and $f(v)=y$. The corresponding decision problem is called {\sc $H$-Compaction}.
If $H={\cal C}_4$, then the problem is equivalent to {\sc Disconnected Cut} 
when restricted to graphs of diameter~$2$~\cite{FMPS09}.
Hence, {\sc ${\cal C}_4$-Compaction} is \NP-complete  for graphs of diameter~$2$~\cite{MP15}
 (the result of~\cite{Vi02} holds for graphs of diameter at least~3).
Similarly, a homomorphism~$f$ from a graph~$G$ to a graph~$H$ is {\it (vertex-)surjective} if for every $x\in V(H)$ there is a vertex~$u\in V(G)$ such that $f(u)=x$. The decision problem is called {\sc Surjective $H$-Colouring} (or $H$-{\sc Vertex Compaction}, or {\sc Surjective $H$-Homomorphism}) and is equivalent to {\sc Disconnected Cut} if $H={\cal C}_4$. The complexity classifications of {\sc $H$-Compaction} and {\sc Surjective $H$-Colouring} are wide open despite many partial results; 
see~\cite{BKM12} for a survey and \cite{GJMPS17} for a more recent overview focussing on {\sc Surjective $H$-Colouring} .

\subsection{Overview}

In Section~\ref{s-known} we state several underlying structural observations for graphs of diameter~$2$. In Sections~\ref{s-circular} and~\ref{s-linegraphs}, respectively, we prove that {\sc Disconnected Cut} can be solved in polynomial time for circular-arc graphs and line graphs, respectively. In Section~\ref{s-claw} we prove our main result, and in particular, our new structural theorem for claw-free graphs of diameter~$2$. In Section~\ref{s-k4} we show that {\sc Disconnected Cut} is polynomial-time solvable on paw-free graphs, co-paw-free graphs and on diamond-free graphs. By combining these results with our result for claw-free graphs we prove that
 {\sc Disconnected Cut} is polynomial-time solvable for $H$-free
 graphs whenever $H$ is a graph on at most four vertices not isomorphic to the complete graph $K_4$. We pose the case where $H=K_4$ as an open problem in Section~\ref{s-con}, together with some other relevant open problems.

\section{Preliminaries and Basic Results}\label{s-known}

In the remainder of our paper,  graphs are finite, undirected, and have neither multiple edges nor self-loops unless explicitly stated otherwise.

Let $G=(V,E)$ be a graph.  For a set $S\subseteq V$, $G[S]$ is the subgraph of $G$ induced by $S$.
We say that $S$ is {\it connected} if~$G[S]$ is connected. 
We write $G-S=G[V\setminus S]$, and if $S=\{u\}$, we write $G-u$ instead.
For a vertex $u\in V$, let $N(u)=\{v \;|\; uv\in E\}$ be the neighbourhood of~$u$
and $N[u] = N(u) \cup\{u\}$. 
The {\it complement} $\overline{G}$ of $G$ has vertex set $V$ and edge set
$\{uv\; |\; uv\notin E\}$.
The \emph{contraction} of an edge $uv\in E$ is the operation that removes the vertices $u$ and $v$ from $G$, and replaces $u$ and $v$ by a new vertex that is made adjacent to precisely those vertices that were adjacent to $u$ or $v$ in $G$ (without introducing self-loops nor multiple edges).
The {\em distance} $d_G(u,v)$ between two vertices $u$ and $v$ of $G$ is the number of edges in a shortest path between them.
If $u$ and $v$ are in different connected components of $G$, then $d_G(u,v)=\infty$.
The {\em diameter} of $G$ is equal to $\max\{d_G(u,v)\; |\; u,v\in V\}$.

The following lemma was observed by Fleischner et al.

\begin{lemma}[\cite{FMPS09}]\label{l-diameter}
Let $G$ be a graph. If $G$ has diameter~$1$, then $G$ has no disconnected cut.
If $G$ has diameter at least~$3$, then $G$ has a disconnected cut.
\end{lemma}
 
A subset~$D\subseteq V$ is a {\it dominating} set of a graph~$G=(V,E)$ if every vertex of $V\setminus D$ is adjacent to at least one vertex of $D$. 
If $D=\{u\}$, then $u$ is a {\it dominating} vertex of $G$.
An edge $uv$ of a graph $G=(V,E)$ is {\it dominating} if $\{u,v\}$ is dominating.
A vertex~$u\in V$ has a {\it disconnected neighbourhood} if $N(u)$ induces a disconnected graph.

We need the following two lemmas, the first one of which is a straightforward observation.

\begin{lemma}\label{l-dominating}
If a graph~$G$ contains a dominating vertex, then $G$ has no disconnected cut.
\end{lemma}

\begin{lemma}\label{l-disconnected}
If a graph~$G$ contains a non-dominating vertex~$u$ with a disconnected neighbourhood, then $G$ has a disconnected cut.
\end{lemma}

\begin{proof}
Let $A_1,\ldots,A_r$ be the connected components of $G[N(u)]$ for some $r\geq 2$. As $u$ is not dominating, $G-(N(u)\cup \{u\})$ is nonempty. We define $V_1=\{u\}$, $V_2=V(A_1)$, $V_3=V(A_2) \cup \cdots \cup  V(A_r)$ and $V_4=V(G)-(N(u)\cup \{u\})$ and find that $V_1\cup V_3$ (or $V_2\cup V_4$) is a disconnected cut of $G$.
\qed
\end{proof}

Two disjoint vertex sets $S$ and $T$ in a graph~$G=(V,E)$ are {\it complete} if there is an edge between every vertex of $S$ and
every vertex of $T$, and $S$ and $T$ are {\it anticomplete} if there is no edge between a vertex of $S$ and a vertex of $T$.
Recall that $G$ has a {\it disconnected cut} if $V$ can be partitioned into four (nonempty) sets $V_1$, $V_2$, $V_3$, $V_4$, such that $V_1$ is anticomplete to $V_3$ and $V_2$ is anticomplete to $V_4$. We say that $V_1,V_2,V_3,V_4$ form a {\it disconnected partition} of~$G$. 

\begin{lemma}\label{l-4cycle}
Let $V_1$, $V_2$, $V_3$, $V_4$ be a disconnected partition of a graph $G$ of diameter~$2$. Then $G$ has an induced cycle $C$
with $4\leq |V(C)|\leq 5$ such that $V(C)\cap V_i\neq \emptyset$ for $i=1,\ldots,4$.
\end{lemma}

\begin{proof}
Let $u_1\in V_1$ and $u_3\in V_3$. As $G$ has diameter~2, there exists a vertex $u_2$ in $V_2$ or $V_4$, say $V_2$, such
that $u_2$ is adjacent to $u_1$ and to $u_3$. Let $u_4\in V_4$.
As $G$ has diameter~2, there exists a vertex $u_1'$ in  $V_1$ or $V_3$, say $V_1$, such
that $u_1'$ is adjacent to $u_2$ and $u_4$. If $u_3$ and $u_4$ are adjacent, then we can take as $C$ the cycle on vertices 
$u_1'$, $u_2$, $u_3$, $u_4$ in that order. Otherwise, as $G$ has diameter~2, there exists a vertex $w\in V_3\cup V_4$, such
that $w$ is adjacent to $u_3$ and to $u_4$. In that case we can take as $C$ the cycle on vertices $u_1'$, $u_2$, $u_3$, $w$, $u_4$.
\qed
\end{proof}

Two adjacent vertices $u$ and $v$ of graph~$G=(V,E)$ have a {\it nested neighbourhood} if $N(u)\setminus \{v\}\subseteq N(v)\setminus \{u\}$ or $N(v)\setminus \{u\}\subseteq N(u)\setminus \{v\}$.
We say that $G$ has \emph{distinct neighbourhoods} if $G$ has no two vertices that have nested neighbourhoods.
In our proof we will apply the following lemma exhaustively.

\begin{lemma}\label{l-neighbourhood}
Let $G$ be a graph of diameter~$2$ that contains two vertices $u$ and $v$ such that $N(u)\setminus \{v\}\subseteq N(v)\setminus \{u\}$. Then $G$ has a disconnected cut if and only if $G-u$ has a disconnected cut. Moreover, $G-u$ has diameter at most~$2$.
\end{lemma}

\begin{proof}
As $G$ has diameter~2 and $N(u)\setminus \{v\}\subseteq N(v)\setminus \{u\}$, we find that $G-u$ has diameter at most~2.

First suppose that $G$ has a disconnected cut. Then $G$ has a disconnected partition $V_1,V_2,V_3,V_4$. We may assume without loss of generality that $v$ belongs to $V_1$. 
First assume that $u\in V_1$. 
Hence, $\{u,v\} \subseteq V_1$ and thus, 
$V_1\setminus \{u\}, V_2, V_3, V_4$ form a disconnected partition of $G-u$. Hence, $G-u$ has a disconnected cut.
Now assume that $u\notin V_1$. As $u$ is adjacent to $v$, we may assume without loss of generality that $u$ belongs to $V_2$. 
As $v$ belongs to $V_1$ and $V_1$ is anticomplete to $V_3$, we find that $v$ has no neighbours in $V_3$.
As $N(u)\setminus \{v\}\subseteq N(v)\setminus \{u\}$, this means that $u$ has no neighbours in $V_3$. As $G$ has diameter~2, this means that $V_2$ contains a vertex~$w\neq u$ that has a neighbour in $V_3$, thus $V_2\setminus \{u\}$ is nonempty. As a consequence, $V_1,V_2\setminus \{u\}, V_3, V_4$ form a disconnected partition of $G-u$. Hence, $G-u$ has a disconnected cut.

Now suppose that $G-u$ has a disconnected cut. Then $G-u$ has a disconnected partition $V_1',V_2',V_3',V_4'$. We may assume without loss of generality that $v$ belongs to $V_1'$.
Then, as $N(u)\setminus \{v\}\subseteq N(v)\setminus \{u\}$, we find that $V_1'\cup \{u\}, V_2', V_3', V_4'$ form a disconnected partition of $G$. Hence, $G$ has a disconnected cut.
\qed
\end{proof}

A pair of vertices $u$ and $v$ of a graph $G=(V,E) $ is a {\it universal pair} if $\{u,v\}$ is a dominating set and there exist distinct
vertices $x$ and $y$ in $V\setminus \{u,v\}$, such that $x\in N(u)$ and $y\in N(v)$;
note that this implies that $|V|\geq 4$ and $u,v$ have at least one neighbour in $V-\{u,v\}$.
Let $H$ be a graph.
Then $G$ is {\it $H$-free} if $G$ contains no induced subgraph isomorphic to $H$. 
The {\em disjoint union} $G+\nobreak H$ of two vertex-disjoint graphs~$G$ and~$H$ is the graph $(V(G)\cup V(H), E(G)\cup E(H))$. The disjoint union of~$r$ copies of a graph~$G$ is denoted by~$rG$.
The graphs $C_r$ and $P_r$ denote the cycle and path on~$r$ vertices, respectively. The graph $K_r$ denotes the complete graph on $r$ vertices.
The {\it independence number} $\alpha(G)$ of a graph $G$ is the largest $k$ such that $G$ contains an induced subgraph isomorphic to $kP_1$.

\begin{lemma}[\cite{CDEFFK10}]\label{l-2p2}
A $2P_2$-free graph has a disconnected cut if and only if its complement has a universal pair.
\end{lemma}

\begin{lemma}[\cite{DMS12}]\label{l-4p1}
{\sc Disconnected Cut} is $O(n^3)$-time solvable for $4P_1$-free graphs.
\end{lemma}

A graph $G$ is {\it bipartite} if $V(G)$ can be partitioned into two classes $A$ and $B$ such that every edge of $G$ has an endpoint in $A$ and an endpoint in $B$. If $A$ is complete to $B$, then $G$ is a {\it complete} bipartite graph.  The graph $K_{s,t}$ denotes the complete bipartite graph wit partition classes of size~$s$ and~$t$, respectively.
The graph $(\{u,v_1,v_2,v_2\},\{uv_1,uv_2,uv_3\})$
is the {\it claw} $K_{1,3}$. A {\it cobipartite} graph is the complement of a bipartite graph.

The {\it line graph} of a graph $G$ with edges $e_1,\ldots,e_p$ is the graph
$L(G)$ with vertices $u_1,\ldots,u_p$ such that there is an edge between any
two vertices $u_i$ and $u_j$ if and only if $e_i$ and $e_j$ have a common endpoint in $G$. 
Note that every line graph is claw-free.
We call $G$ the {\it preimage} of $L(G)$.  Every connected line graph except $K_3$ has a unique 
preimage~\cite{Ha69}.

A {\it circular-arc graph} is a graph that
has a representation in which each vertex 
corresponds to an arc of a circle, such that two vertices 
are adjacent if and only if their corresponding arcs intersect.
An {\it interval graph} is a graph that has representation in which each vertex corresponds to an interval of the line, such that two vertices are adjacent if and only if their corresponding intervals intersect. Note that circular-arc graphs generalize interval graphs.
A circular-arc or interval graph is \emph{proper} if it has a representation where the arcs respectively intervals are such that no one is contained in another.
A {\it chordal graph} is a graph in which every induced cycle is a triangle; note that every interval graph is chordal, but that there exist circular-arc graphs, such as cycles, that are not chordal.

\section{Circular-Arc Graphs}\label{s-circular}

In this section we prove that {\sc Disconnected Cut} is polynomial-time solvable for circular-arc graphs. 
This result is known already for interval graphs, as it follows from the result that {\sc Disconnected
Cut} is polynomial-time solvable for the class of chordal graphs~\cite{IKPT11}, which contains the class of interval graphs. In fact, we have
an $O(n^2)$-time algorithm for interval graphs. Due to Lemma~\ref{l-4cycle} and the fact that interval graphs are chordal, no interval graph of diameter~2 has a disconnected cut. Consequently, an interval graph has a disconnected cut if and only if its diameter is at least~3 due to Lemma~\ref{l-diameter}. To show that {\sc Disconnected Cut} is polynomial-time solvable for circular-arc graphs requires significant additional work. 

Let $G$ be a circular-arc graph. For each vertex $u\in V(G)$ we can associate an arc $[l_u,r_u]$ where we say that $l_u$ is the clockwise left endpoint of $u$ and $r_u$ is the clockwise right endpoint of $u$. 
The following result of McConnell shows that we may assume that all left and right endpoints of the vertices of $G$ are unique.

\begin{lemma}[\cite{Mc03}]\label{l-mcc}
A circular-arc graph $G$ on $n$ vertices and $m$ edges can be recognized in $O(n+m)$ time. In the same time, a representation of $G$ can be constructed with distinct arc endpoints that are
clockwise enumerated as $1,\ldots,2n$.   
\end{lemma}

We need the following lemma.

\begin{lemma}\label{l-circular}
Let $G$ be a circular-arc graph of diameter~$2$ with a disconnected cut. Then $G$ has a disconnected partition $V_1$, $V_2$, $V_3$, $V_4$ such that each $V_i$ is connected.
\end{lemma}

\begin{proof}
By Lemma~\ref{l-mcc} we may assume that $G$ has a  representation with distinct arc endpoints clockwise enumerated as $1,\ldots,2n$.
Let $V_1$, $V_2$, $V_3$, $V_4$ be a disconnected partition of $G$.
By Lemma~\ref{l-4cycle}, $G$ contains a cycle~$C$
 with vertices $u_i$ for $i=1,\ldots, j$ (with $u_{j+1}=u_1$) and $j\in \{4,5\}$,
such that $V(C)\cap V_i\neq \emptyset$ for $i=1,2,3,4$. We may assume without loss of generality that if $j=5$, then $u_5\in V_4$.

Let~$D_i$ be the connected component of $G[V_i]$ that contains $u_i$ for $i=1,\ldots,4$. Note that if $j=5$, then
$u_5\in D_4$.
We let $[l_i,r_i]$ denote the arc corresponding to the arc covered by the vertices of $D_i$ (or equivalently, the arc associated with 
the vertex obtained after contracting the connected graph~$D_i$ into a single vertex). 
Note that the arcs $[l_1,r_1]$, \ldots, $[l_4,r_4]$ cover the whole circle,
as they contain the arcs corresponding to the vertices of $C$.
Moreover,
$[l_i,r_i]$ intersects with $[l_{i-1},r_{i-1}]$ and $[l_{i+1},r_{i+1}]$ for $i=1,\ldots,4$ (where $[l_0,r_0]=[l_5,r_5]=[l_1,r_1]$).

We set $V_i':=V(D_i)$. If $V_i'=V_i$ for $i=1,\ldots, 4$, then $V_1$, $V_2$, $V_3$, $V_4$ is our desired partition of $G$. 
Assume that this is not the case. Then $V_1'\cup \cdots \cup V_4'\subsetneq V$. Consider an arbitrary vertex~$v\in V(G)$ that is not in any~$V_i'$. Let $p$ be 
the number of arcs $[l_i,r_i]$ that intersect with $[l_v,r_v]$. We claim that $p\leq 1$.

First suppose that $p=4$. Then $v$ would be adjacent to a vertex of every $V_i$. This is not possible, because $V_1$, $V_2$, $V_3$, $V_4$ is a disconnected partition of $G$.
Now suppose that $p=3$, say $[l_v,r_v]$ intersects with $[l_1,r_1]$, $[l_2,r_2]$ and $[l_3,r_3]$. Then $v$ must be in $V_2$, that is, $v$ is in the same set
as the vertices of $D_2$. As $v$ intersects with $[l_2,r_2]$ we would have put $v$ in $D_2$, a contradiction.
Now suppose that $p=2$. By construction, $[l_v,r_v]$ intersects with two consecutive arcs, say with arcs $[l_1,r_1]$ and $[l_2,r_2]$.
Then $v\in V_1$ or $v\in V_2$. If $v\in V_1$ then we would have put $v$ in $D_1$, and if $v\in V_2$ then we would have put
$v\in D_2$. Hence, this situation is not possible either.

From the above, we conclude that $p\leq 1$, that is,
for each $v\notin V_1'\cup \cdots V_4'$, we have that $[l_v,r_v]$ intersects with at most one~$[l_i,r_i]$. As 
the arcs $[l_1,r_1],\ldots,[l_4,r_4]$ cover the whole circle, this means that $[l_v,r_v]$ is contained in some $[l_i,r_i]$, and we can 
safely put $v$ in $V_i'$ without destroying the connectivity or the anticompleteness of $V_1,V_3$, or of $V_2,V_4$.
Afterwards, the
four sets $V_1'$, $V_2'$, $V_3'$, $V_4'$   form a disconnected partition of $G$, such that each $V_i'$ is connected.
\qed
\end{proof}

We are now ready to prove the main result of this section.

\begin{theorem}\label{t-circulararc}
{\sc Disconnected Cut} is $O(n^2)$-time solvable for circular-arc graphs.
\end{theorem}

\begin{proof}
Let $G=(V,E)$ be a circular-arc graph on $n$ vertices. We compute the diameter of $G$ in $O(n^2)$ time using the 
(more general) $O(n^2)$-time 
algorithm of~\cite{CLSS98} (or of~\cite{ACL95} or the linear-time algorithm of~\cite{MPP2007}).
Lemma~\ref{l-diameter} tells us that if $G$ has diameter~1, then $G$ has no disconnected cut, and if $G$ has diameter at least~3, then $G$ has a disconnected cut.
Assume that $G$ has diameter~$2$. Lemma~\ref{l-circular} tells us that if $G$ has a disconnected cut, then $G$ has a disconnected partition $V_1$, $V_2$, $V_3$, $V_4$ such that $V_i$ is connected for $i=1,\ldots,4$. 
We say that the \emph{arc of a set $V_i$} is the union of all the arcs of the vertices in $V_i$. As $G$ has diameter~2, the union of the arcs of the sets $V_i$ cover the whole circle.
Moreover, the arcs of $V_1$ and $V_3$ are disjoint and the arcs of $V_2$ and $V_4$ are disjoint.

We now compute, in linear time, a representation of $G$ with distinct arc endpoints clockwise enumerated as $1,\ldots,2n$ via Lemma~\ref{l-mcc}. 
After sorting the arcs in $O(n\log n)$ time, we apply the following procedure for each $v\in V$.
We take a neighbour $v'$ with right-most right endpoint amongst all neighbours of $v$.
We then take a neighbour $v''$ with right-most right endpoint amongst all neighbours of $v'$.
We check if $v''$ is also a neighbour of $v$. 
If so, then $v,v',v''$ form a triangle if $v\neq v''$ and an edge if $v=v''$ such that the corresponding arcs cover the whole circle.
In that case $G$ has no disconnected partition $V_1$, $V_2$, $V_3$, $V_4$ such that $V_i$ is connected for $i=1,\ldots,4$
(as the arcs of the three vertices of the triangle must be placed in the corresponding arcs of the sets $V_i$), and then, by  Lemma~\ref{l-circular}, we find that $G$ has no disconnected cut. 
Note that this procedure takes $O(n^2)$ time in total.

Suppose $G$ has no pair or triple of vertices whose arcs cover the whole circle. 
We perform the following procedure exhaustively. We pick a vertex $v_1$ of $G$ and take a neighbour $v_2$ with right-most right endpoint amongst all neighbours of $v_1$. We then do the same for $v_2$, and so on.
If this procedure ends without yielding an induced cycle, then $G$ has no induced cycle on four or more vertices. Hence, 
$G$ has no disconnected cut due to Lemma~\ref{l-4cycle}. 
Otherwise, we have found in $O(n)$ time, an induced cycle $C$
with vertices $v_1,\ldots,v_k$, in that order, for some $k\geq 4$. As $G$ has diameter~2, we find that $k\in \{4,5\}$. Moreover, 
by construction and because $G$ is circular-arc, the arcs corresponding to the vertices $v_i$ must cover the whole circle. 

If $G$ has a disconnected partition $V_1$, $V_2$, $V_3$, $V_4$ such that $V_i$ is connected for $i=1,\ldots,4$, then the above implies the following. If $k=4$, we may assume without loss of generality that $v_i\in V_i$ for $i=1,\ldots,4$.
If $k=5$, two vertices $v_i,v_{i+1}$ belong to the same set $V_h$, whereas the other sets $V_i$ with $i \not= h$ each contain a single vertex from $C$. If $k=5$, then we guess which two vertices $v_i,v_{i+1}$ will be put in the same set, say $v_1,v_5$; this does not influence the asymptotic running time.
Now we build up the sets $V_i$ from scratch by putting in the vertices from $\{v_1,\ldots,v_k\}$.

We will always maintain that each $V_i$ induces a connected graph, and thus, the union of the arcs of the vertices in $V_i$ indeed always form an arc.
We say that a vertex $u$ \emph{intersects} a set $V_i$ if the arc of $u$ intersects the arc of $V_i$. Note that, since the arcs corresponding to $\{v_1,\ldots,v_k\}$ cover the entire circle, so do the arcs of the sets $V_i$ that we are constructing. 
If there is a vertex that intersects each of the sets $V_i$ constructed so far, then there is no disconnected cut with each $V_i$ connected. If $k=4$, this mean that $G$ has no disconnected cut due to Lemma~\ref{l-circular}. If $k=5$, our guess of vertices $v_1,v_5$ to belong to $V_1$ may have been incorrect, and we need to put two other consecutive vertices of $C$ in the same set $V_h$ before concluding that $G$ has no disconnected cut.

Otherwise, we do as follows.
Note that any vertex $u$ that intersects two sets $V_i$ and $V_{i+2}$ for some $i$ (say, $i \in \{1,2\}$ without loss of generality), also intersects $V_{i+1}$ or $V_{i+3}$ (where $V_5 = V_1$). 
We now put a vertex $u$ that intersects two sets $V_i$ and $V_{i+2}$ for some $i$ into set $V_{i+1}$ if $u$ intersects $V_{i+1}$ as well; otherwise, $u$ intersects $V_{i+3}$ and we put $u$ in $V_{i+3}$.

Let $T$ be the set of vertices of $G$ that we have not placed in some set $V_i$ yet. We claim that each vertex of $T$ must intersect with exactly two sets $V_i$ and $V_j$ such that, in addition, $j=i+1$ holds. 
For contradiction, assume this does not hold for $u\in T$. Then $u$ intersects exactly one set $V_i$, say $V_1$. There must exist a path from $u$ to $v_3$ of length~2, as $G$ has diameter~2. Hence, there exists a vertex $w$ that is adjacent to both $u$ and $v_3$. This means that the arc corresponding to $w$ must intersect $V_1$ and $V_3$. Hence, $w$ has already been placed in $V_2$ or $V_4$ and thus the arc of $u$ intersects two sets, namely $V_1$ and $V_2$, or $V_1$ and $V_4$, a contradiction.
The claim follows.

As every vertex in $T$ intersects two sets $V_i$ and $V_{i+1}$, we can model the remaining instance as an instance of {\sc 2-Satisfiability} as follows. Let $u\in T$, where $u$ must be placed in, say, $V_i$ or $V_{i+1}$. We introduce two variables 
$x_u^i$ and $x_u^{i+1}$ with clauses $x_u^i\vee x_u^{i+1}$ and $\bar{x}_u^i\vee \bar{x}_u^{i+1}$.
For each edge $uv$ where $u$ must be placed in $V_i$ or $V_{i+1}$ and $v$ must be placed in $V_{i+1}$ or $V_{i+2}$, we introduce the clause $\bar{x}_u^i\vee \bar{x}_v^{i+2}$. 
For each edge $uv$ where $u$ must be placed in $V_i$ or $V_{i+1}$ and $v$ must be placed in $V_{i+2}$ or $V_{i+3}$, we  introduce the clauses $\bar{x}_u^i\vee \bar{x}_v^{i+2}$ and $\bar{x}_u^{i+1}\vee \bar{x}_v^{i+3}$. 

The correctness of our algorithm follows from the above description.
Assigning the vertices to the sets $V_i$ takes $O(n^2)$ time, whereas solving the corresponding instance of {\sc 2-Satisfiability} takes $O(n^2)$ time as well.
As computing the diameter takes $O(n^2)$ time and all other steps take $O(n^2)$ time as well, the total running time is $O(n^2)$.
\qed
\end{proof}

\section{Line Graphs}\label{s-linegraphs}

In this section we prove that {\sc Disconnected Cut} is polynomial-time solvable for line graphs. 
We start with the following lemma due to Ito et al.~\cite{IKPT11}.

\begin{lemma}[\cite{IKPT11}]\label{l-linegraph}
Let  $G$ be a graph with diameter~$2$ whose line graph $L(G)$ also has diameter~$2$. 
Then $G$ has a disconnected cut if and only if $L(G)$ has a disconnected cut.
\end{lemma}

We need a lemma on graphs whose line graph has diameter~2.

\begin{lemma}\label{l-2k2}
Let $G$ be a graph that is neither a triangle nor a star. Then $L(G)$ has diameter~$2$ if and only if $G$ is $2P_2$-free.
\end{lemma}

\begin{proof}
Let $G$ be a graph that is neither a triangle nor a star.
First suppose that $L(G)$ has diameter 2. In order to obtain a contradiction, assume that
$G$ is not $2P_2$-free. Then $G$ contains an induced subgraph $H$ with vertices $s,t,u,v$ and edges $e_1=st$ and $e_2=uv$.
As $H$ is an induced subgraph of $G$, we find that $e_1$ and $e_2$ are non-adjacent vertices in $L(G)$. Then, because $L(G)$ has diameter 2, $L(G)$ contains a vertex $e_3$ that is adjacent to $e_1$ and $e_2$. However, then $e_3$ is an edge with one endvertex in $\{s,t\}$ and the other one in $\{u,v\}$. This means that $H$ is not induced, a contradiction.

Now suppose that $G$ is $2P_2$-free. In order to obtain a contradiction, assume that $L(G)$ has diameter not equal to~2.
If the diameter of $L(G)$ is~1, then $L(G)$ is a complete graph implying that $G$ is a triangle or star, which is not what we assume.
If the diameter of $L(G)$ is at least~3, then $L(G)$ contains two vertices $e_1$ and $e_2$ that are of distance at least~3. 
This means that $e_1$ and $e_2$ form an induced $2P_2$ in $G$, a contradiction.
\qed
\end{proof}

We are now ready to prove the main result of this section.

\begin{theorem}\label{t-linegraphs}
{\sc Disconnected Cut} is $O(n^4)$-time solvable for line graphs of $n$-vertex~graphs.
\end{theorem}

\begin{proof}
Let $G$ be a graph on $n$ vertices and $m$ edges. 
We will show how to decide in $O(n^4)$ time if $L(G)$ has a disconnected cut.
We first check in $O(n)$ time if $G$ is a triangle or star.
If so, then $L(G)$ is a complete graph and thus $L(G)$ has no disconnected cut.
From now on suppose that $G$ is neither a triangle nor a star.
By Lemma~\ref{l-2k2} we find that $L(G)$ has diameter~2 if and only if $G$ is $2P_2$-free. Hence we can check in 
$O(n^4)$ time, via checking if $G$ has an induced $2P_2$ by brute force, if $L(G)$ has diameter~2. 

First assume that $L(G)$ does not have diameter~2. As $G$ is not a triangle or a star, $L(G)$ has diameter at least~3.
By Lemma~\ref{l-diameter} we find that $L(G)$ has a disconnected cut.
Now assume that $L(G)$ has diameter~2.
We check in $O(n^3)$ time if $G$ has an edge $uv$ such that every vertex of $V(G)\setminus \{u,v\}$ is adjacent to at least one of
$u,v$. If so, then $uv$ is a dominating vertex of $L(G)$, and $L(G)$ has no disconnected cut due to Lemma~\ref{l-dominating}.
If not, then $L(G)$ has no dominating vertices, and we proceed as follows. First we check if $L(G)$ has a vertex $uv$ with a disconnected neighbourhood, or equivalently, if $G$ contains an edge $uv$ such that $u$ and $v$ have degree at least~2 and no common neighbours.
This takes $O(n^3)$ time. If $L(G)$ has a vertex with a disconnected neighbourhood, then $L(G)$ has a disconnected
cut by Lemma~\ref{l-disconnected}.
From now on assume that $L(G)$ has no vertex with a disconnected neighbourhood.
As $G$ is neither a triangle nor a star, $G$ is $2P_2$-free by Lemma~\ref{l-2k2}.
Hence, $G$ has diameter at most~3.  
We can determine in $O(n^3)$ time the diameter of $G$ and consider each case separately.

\medskip
\noindent
{\bf Case 1.} $G$ has diameter~1.\\
We claim that $L(G)$ has no disconnected cut. For contradiction, assume that $L(G)$ has a disconnected cut. Let $V_1'$, $V_2'$, $V_3'$, $V_4'$ be a disconnected partition of $L(G)$.
By Lemma~\ref{l-4cycle}, $L(G)$ contains a cycle~$C'$
 with vertices $u_iu_{i+1}$ for $i=1,\ldots, j$ (with $u_{j+1}=u_1$) and $j\in \{4,5\}$,
such that $V(C')\cap V_i'\neq \emptyset$ for $i=1,2,3,4$. 
Then we may assume without loss of generality that $u_iu_{i+1}\in V_i'$ for $i=1,\ldots,4$ and $u_ju_{j+1}\in V_4'$.
As $G$ has diameter~1, $u_1u_3$ is an edge of $G$ and thus a vertex of $L(G)$. In $L(G)$, $u_1u_3$ is adjacent 
to every vertex in $\{u_1u_2,u_2u_3,u_3u_4,u_ju_{j+1}\}$, and thus to a vertex in $V_i'$ for $i=1,\ldots,4$, a contradiction.

\medskip
\noindent
{\bf Case 2.} $G$ has diameter~2.\\
Then $G$ has a disconnected cut if and only if $L(G)$ has a disconnected cut due to Lemma~\ref{l-linegraph}.
By Lemma~\ref{l-2p2} it suffices to check if $\overline{G}$ has a universal pair. This takes $O(n^3)$ time.

\medskip
\noindent
{\bf Case 3.} $G$ has diameter~3.\\
We will prove that $L(G)$ has no disconnected cut.
As $G$ has diameter~3, $G$ does have a disconnected cut by Lemma~\ref{l-diameter}. We need the following claim.

\medskip
\noindent
{\it Claim.
Let $V_1, V_2, V_3, V_4$ be a disconnected partition of $G$. Then every cycle~$C$ of $G$ with $4\leq |V(C)|\leq 5$ contains
vertices of at most three distinct sets from $\{V_1, V_2, V_3, V_4\}$.}

\medskip
\noindent
We prove the Claim as follows. For contradiction, assume that $G$ has a cycle $C$ with vertices $u_1,\ldots,u_j$ for $j\in \{4,5\}$,
such that $V(C)\cap V_i\neq \emptyset$ for $i=1,\ldots,4$. We may assume without loss of generality that $u_i\in V_i$ for
$i=1,\ldots,4$ and $u_j\in V_4$. 
As $G$ is $2P_2$-free, we may assume without loss of generality that $u_3$ is in a singleton connected component of 
$G[V_3]$. If $j=4$, then we may also assume without loss of generality that $u_2$ is in a singleton connected component of $G[V_2]$.
If $j=5$, then $u_2$ must be in a singleton connected component of $G[V_2]$ due to the edge $u_4u_5$, which is contained
in $G[V_4]$. 
This means that the sets $E(u_2)=\{u_2w \; |\; w\in N_G(u_2)\setminus \{u_3\}\}$ and $E(u_3)=\{u_3w\; |\; w\in N_G(u_3)\setminus \{u_2\}\}$ are disjoint.
As $u_1u_2$ and $u_3u_4$ are edges of $G$, both $E(u_2)$ and $E(u_3)$ are nonempty.
Hence, the vertex $u_2u_3$ has a disconnected neighbourhood in $L(G)$, a contradiction. This proves the Claim.

\medskip
\noindent
Now, for contradiction, assume that $L(G)$ has a disconnected cut. 
 Let $V_1'$, $V_2'$, $V_3'$, $V_4'$ be a disconnected partition of $L(G)$. By Lemma~\ref{l-4cycle}, $L(G)$ contains a cycle~$C'$
 with vertices $u_iu_{i+1}$ for $i=1,\ldots, j$ (with $u_{j+1}=u_1$) and $j\in \{4,5\}$,
such that $V(C')\cap V_i'\neq \emptyset$ for $i=1,2,3,4$. 
Then we may assume without loss of generality that $u_iu_{i+1}\in V_i'$ for $i=1,\ldots,4$ and $u_ju_{j+1}\in V_4'$. 

We define the following partition $V_1$, $V_2$, $V_3$, $V_4$ of $V(G)$. Let $u\in V(G)$. If $u$ is incident to only edges from one set  $V_i'$, then we put $u$ in $V_i$. Suppose $u$ is incident to edges from more than one set $V_i'$. As $V_1'$, $V_2'$, $V_3'$, $V_4'$ is a disconnected partition of $L(G)$, we find that $u$ is incident to edges from $V_i'$ and $V_{i+1}'$ for some $1\leq i\leq 4$
(where $V_5=V_1$) and to no other sets $V_j'$. In that case we put $u$ into $V_{i+1}$. 

We now prove that $V_1$ is anticomplete to $V_3$. For contradiction, suppose that $V_1$ contains a vertex $u$ and $V_3$ contains a vertex $v$ such that $uv\in E(G)$. As $u\in V_1$, we find that $u$ is incident to edges only in $V_4'$ and $V_1'$. Hence, $uv\in V_1'\cup V_4'$. As $v\in V_3$, we find that $v$ is incident to edges only in $V_2'$ and $V_3'$. Hence $uv\in V_2'\cup V_3'$, a contradiction.
By the same argument we can show that $V_2$ is anticomplete to $V_4$.
Let $C$ be the cycle with vertices $u_1,\ldots, u_j$ in $G$. Then $V(C)\cap V_i\neq \emptyset$, and thus $V_i\neq \emptyset$, for $i=1,\ldots,4$.
Hence, $V_1$, $V_2$, $V_3$, $V_4$ is a disconnected partition of $G$, and $C$ is a cycle in $G$ with $V(C)\cap V_i\neq \emptyset$
for every $i$.
This is not possible due to the Claim. We conclude that $L(G)$ has no disconnected cut.

\smallskip
\noindent
The correctness of our algorithm follows from the above description. 
If $G$ has diameter~1 (Case 1) or diameter~3 (Case 3), no additional running time is required, as 
we showed that 
$L(G)$ has no disconnected cut in both these cases. Hence,
only executing Case~2 takes additional time, namely time $O(n^3)$.
Hence the total running time of our algorithm is~$O(n^4$), as desired.
\qed
\end{proof}

\section{Claw-Free Graphs}\label{s-claw}
In this section, we prove that {\sc Disconnected Cut} is polynomial-time solvable on claw-free graphs. The proof consists of two parts. In Section~\ref{s-1} we show how to get rid of certain cobipartite structures in the graph, called W-joins. We remark that {\sc Disconnected Cut} can be solved in polynomial time on cobipartite graphs~\cite{FMPS09}. Although this is a necessary condition for {\sc Disconnected Cut} to be solvable in polynomial time on claw-free graphs, the algorithm for cobipartite graphs is not sufficient to deal with W-joins.
In Section~\ref{s-2} we present our new decomposition theorem for claw-free graphs of diameter~2 and
combine this theorem with the results from the previous sections and Section~\ref{s-1} to show our main result.

\subsection{Cobipartite Structures versus Disconnected Cuts}\label{s-1}

We consider the following cobipartite structures that might be present in claw-free graphs~\cite{CS05,HMLW11-arxiv,HMLW11}. 
A pair $(A,B)$ of disjoint non-empty sets of vertices is a \emph{W-join} in graph $G$ if $|A|+|B| > 2$, $A$ and $B$ are cliques, $A$ is neither complete nor anticomplete to $B$, and every vertex of $V(G) \setminus (A \cup B)$ is either complete or anticomplete to $A$ and either complete or anticomplete to $B$. A W-join is a \emph{proper W-join} if each vertex in $A$ is neither complete nor anticomplete to $B$ and each vertex in $B$ is neither complete nor anticomplete to $A$. Observe that for a proper W-join $(A,B)$, it must hold that $|A|,|B| \geq 2$. For any W-join $(A,B)$, it holds that $G[A \cup B]$ is a cobipartite induced subgraph in $G$.

We assume that an input graph $G$ of {\sc Disconnected Cut} has diameter~$2$ and that $G$ has distinct neighbourhoods, by Lemmas~\ref{l-diameter} and~\ref{l-neighbourhood} respectively. We show how to use these assumptions to remove all W-joins in a claw-free graph and obtain an equivalent instance of {\sc Disconnected Cut}.
As a first step, we show that we can focus on proper W-joins.

\begin{lemma}\label{l-mustbeproper}
Let $G$ be a graph with distinct neighbourhoods. If $G$ admits a W-join $(A,B)$, then $(A,B)$ is a proper W-join.
\end{lemma}

\begin{proof}
We need to show that no vertex of $A$ (respectively $B$) is complete or anticomplete to $B$ (respectively $A$). Suppose there exists a vertex $a \in A$ such that $a$ is anticomplete to $B$. Since $A$ is not anticomplete to $B$, there exists a vertex $a' \in A \setminus\{a\}$ such that $a'$ is adjacent to some vertex of $B$. This implies that $N(a)\setminus\{a'\} \subseteq N(a') \setminus\{a\}$, which is a contradiction to the fact that $G$ has distinct neighbourhoods. Hence, no vertex of $A$ is anticomplete to $B$. Similarly, no vertex of $B$ is anticomplete to $A$.

Suppose there exists a vertex $a \in A$ such that $a$ is complete to $B$. Then for every $a' \in A \setminus\{a\}$, it holds that $N(a')\setminus\{a\} \subseteq N(a) \setminus\{a'\}$. Since $G$ has distinct neighbourhoods, no such $a'$ exists, and thus $|A| = 1$. Hence, $|B| \geq 2$ by the definition of a W-join. However, no vertex of $B$ is anticomplete to $A$. Since $|A| = 1$, this implies that each vertex of $B$ is complete to $A$. Following the same reasoning as before, this implies that $|B| = 1$, a contradiction. Hence, no vertex of $A$ is complete to $B$. Similarly, no vertex of $B$ is complete to $A$.
\qed
\end{proof}

We now argue that we can focus on a specific type of proper W-joins.
A W-join $(A,B)$ is \emph{partitionable} if there are partitions of $A$ into non-empty sets $A',A''$ and of $B$ into non-empty sets $B',B''$ such that $A'$ is anticomplete to $B''$ and $B'$ is anticomplete to $A''$. A proper W-join $(A,B)$ is \emph{shatterable} if it is partitionable with sets $A',A'',B',B''$ and one of $(A',B'), (A'',B'')$ is also a proper W-join; we say it is \emph{unshatterable} otherwise.

\begin{lemma}\label{l-shatterc4}
Let $G$ be a graph with distinct neighbourhoods and let $(A,B)$ be a proper W-join in $G$. If $(A,B)$ is partitionable and unshatterable, then $G[A \cup B]$ is isomorphic to $C_4$.
\end{lemma}

\begin{proof}
Suppose $(A,B)$ is partitionable with sets $A',A'',B',B''$. By definition, we have $|A'|, |A''|, |B'|, |B''| \geq 1$. Without loss of generality, $|A'| + |B'| \geq |A''| + |B''|$. Since $(A,B)$ is unshatterable, $(A',B')$ is not a proper W-join. 
Suppose $|A'|+|B'| > 2$. Then, as $G$ has distinct neighbourhoods, $(A',B')$ is a W-join. Consequently, $(A',B')$ is a proper W-join by Lemma~\ref{l-mustbeproper}, which is a contradiction. Hence, $|A'|=|B'| =1$. 
It follows that $|A'|=|B'|=|A''|=|B''| = 1$. Since $(A,B)$ is proper, $G[A \cup B]$ is isomorphic to $C_4$.
\qed
\end{proof}

\begin{lemma}\label{l-unshatterable}
Let $G$ be a claw-free graph that is not cobipartite, has distinct neighbourhoods, and has diameter~$2$. Let $(A,B)$ be a proper W-join in $G$ that is unshatterable. If $G$ admits a disconnected cut, then there exists a disconnected partition $V_1,V_2,V_3,V_4$ of $G$ such that $V_i \cap (A \cup B) = \emptyset$ for some $i \in \{1,2,3,4\}$.
\end{lemma}

\begin{proof}
Let $V_1,V_2,V_3,V_4$ be any disconnected partition of $G$ and suppose that $V_i \cap (A \cup B) \not= \emptyset$ for each $i \in \{1,2,3,4\}$. Note that $V_1$ is anticomplete to $V_3$ and $V_2$ is anticomplete to $V_4$ by definition. 
Then we may assume without loss of generality that $V_1 \cap A \not= \emptyset$ and $V_2\cap A\not= \emptyset$. The former implies that $V_3 \cap (A \cup B) \subseteq B$ and thus $V_1 \cap (A\cup B) \subseteq A$. The latter implies that $V_4 \cap (A \cup B) \subseteq B$ and thus $V_2 \cap (A\cup B) \subseteq A$. It follows that $(A,B)$ is partitionable with sets $V_1 \cap A, V_2 \cap A, V_3 \cap B, V_4 \cap B$. Then, by Lemma~\ref{l-shatterc4} and the assumption that $(A,B)$ is unshatterable, $G[A\cup B]$ is isomorphic to a $C_4$. For $i \in\{1,2,3,4\}$, let $v_i$ be the single vertex in $V_i \cap (A \cup B)$, 
so, for $i\in \{1,2,3,4\}$, $v_i$ is adjacent to $v_{i+1}$ with $v_5=v_1$, and moreover, $A=\{v_1,v_2\}$ and $B=\{v_3,v_4\}$.

Let $P = N(A) \setminus N[B]$, $Q = N(B) \setminus N[A]$, $M = N[A \cup B] \setminus (P \cup Q)$, and $R = V(G) \setminus (P \cup Q \cup M)$.
Note that $P$ is complete to $A$ and anticomplete to $B$, whereas $Q$ is complete to $B$ and anticomplete to $A$. Moreover, $M$ is complete to $A\cup B$, whereas $R$ is anticomplete to $A\cup B$.
This means that $P \subseteq V_1 \cup V_2$, because its neighbourhood includes $v_1$ and $v_2$. Moreover, $Q \subseteq V_3 \cup V_4$, because its neighbourhood includes $v_3$ and $v_4$, whereas $M = \emptyset$, because its neighbourhood includes $v_1$, $v_2$, $v_3$, and $v_4$. Moreover, $G[P]$ is a clique, because two non-adjacent vertices $u,v \in P$ with $v_1$ and $v_4$ yield a claw. Similarly, $G[Q]$ is a clique.

Observe that one of the following moves yields the requested disconnected partition of $G$:
\begin{itemize}
\item moving $v_1$ to $V_2$ and $v_4$ to $V_3$, unless $|V_1| = 1$ or $|V_4| = 1$;
\item moving $v_2$ to $V_1$ and $v_3$ to $V_4$, unless $|V_2| = 1$ or $|V_3| = 1$.
\end{itemize}
The crux is to show that $|V_1|, |V_4|> 1$ or $|V_2|, |V_3|> 1$. We now go through several cases.

Suppose that $R = \emptyset$. Since $M = \emptyset$, it follows that $P \cup A$ and $Q \cup B$ forms a cobipartition of $G$, a contradiction. Hence, $R \not= \emptyset$.

Suppose that $P = \emptyset$. Since $R \not= \emptyset$, let $v \in R$. Note that every path from $v$ to $v_1$ must intersect both $Q$ and $B$, two disjoint sets. This contradicts the assumption that $G$ has diameter~$2$. Hence, $P \not= \emptyset$. Similarly, we derive that $Q \not= \emptyset$.

Without loss of generality, assume that $P \cap V_1 \not= \emptyset$ and let $x \in P \cap V_1$. Suppose that $|V_4| = 1$, and thus $V_4 = \{v_4\}$. Then $Q \subseteq V_3$. Now note that, because $R \not= \emptyset$, there exists a $j \in \{1,2,3\}$ such that $R \cap V_j \not= \emptyset$. We consider all three cases.

Suppose that $R \cap V_1 \not= \emptyset$. Since $V_1$ is anticomplete to $V_3$, any path from a vertex in $R \cap V_1$ to $v_4$ must contain either at least one internal vertex in $R \cup P$ and one in $Q$, or at least one
internal vertex in $P$ in addition to $v_1$. This contradicts the assumption that $G$ has diameter~$2$.

Suppose that $R \cap V_2 \not= \emptyset$. Then $|V_2| > 1$. Since $Q \not=\emptyset$ and $Q \subseteq V_3$, it follows that $|V_3| > 1$. Then moving $v_2$ to $V_1$ and $v_3$ to $V_4$ yields a disconnected partition as requested.

Suppose that $R \cap V_3 \not= \emptyset$. Let $v \in R \cap V_3$. Then a shortest path from $v$ to $v_1$, which has length~$2$, cannot intersect both $Q$ and $B$. Hence, any such shortest path must intersect~$P$. Since $V_3$ is anticomplete to $V_1$, it follows that $V_2 \cap P \not= \emptyset$. Hence, $|V_2|> 1$ and $|V_3|>1$, and then moving $v_2$ to $V_1$ and $v_3$ to $V_4$ yields a disconnected partition as requested.

We may thus assume that $|V_4| > 1$. Since $|V_1| > 1$, by moving $v_1$ to $V_2$ and $v_4$ to $V_3$, we obtain a disconnected partition as requested.
\qed
\end{proof}

Let $(A,B)$ be a proper W-join of a graph $G$.
For any two adjacent vertices $a \in A$ and $b \in B$, let $G_{ab}$ be the graph obtained from $G$ by removing $A \setminus\{a\}$ and $B \setminus\{b\}$.  Observe that the graph $G_{ab}$ is the same regardless of the choice of $a,b$.

\begin{lemma}\label{l-noproper}
Let $G$ be a 
claw-free graph that is not cobipartite, has distinct neighbourhoods, and has diameter~$2$. Let $(A,B)$ be a proper W-join of $G$ that is unshatterable. Then $G$ admits a disconnected cut if and only if $G_{ab}$ admits a disconnected cut for any two adjacent vertices $a \in A$ and $b \in B$.
\end{lemma}

\begin{proof}
First suppose that $G_{ab}$ admits a disconnected partition $V_1,V_2,V_3,V_4$ for any two vertices $a,b$. Let $a \in V_i$ and $b \in V_j$ for $i,j \in \{1,2,3,4\}$. Then the sets $V_1',V_2',V_3',V_4'$ obtained from $V_1,V_2,V_3,V_4$ by adding $A$ to $V_i$ and $B$ to $V_j$ is a disconnected partition of $G$.

Now suppose that $G$ admits a disconnected cut. Let $V_1,V_2,V_3,V_4$ be a disconnected partition of $G$. By Lemma~\ref{l-unshatterable}, we may assume without loss of generality that $V_4 \cap (A \cup B) = \emptyset$. Note that $A$ is a clique in $G$ and $V_1$ is anticomplete to $V_3$, and thus $A \subseteq V_1 \cup V_2$ or $A \subseteq V_2 \cup V_3$. We assume the former without loss of generality. Among all such disconnected partitions, we will assume that $V_1,V_2,V_3,V_4$ was chosen to minimize $|A \cap V_1|$. 

We consider several cases. In each of these cases we find two vertices $a,b$ for which we can construct a disconnected partition of $G_{ab}$.
Note that this suffices to prove the statement, as the graph $G_{ab}$ is the same regardless of the choice of $a,b$.

First assume that $A \subseteq V_1$. Since no vertex of $B$ is anticomplete to $A$ by the definition of a proper W-join and $V_1$ is anticomplete to $V_3$, it follows that $B \subseteq V_1 \cup V_2$. Now if $B \subseteq V_1$, then let $a \in A$ and $b\in B$ be arbitrary adjacent vertices (these exist by the definition of a W-join) and $V_1 \setminus ((A \setminus\{a\}) \cup (B \setminus\{b\})),V_2,V_3,V_4$ is a disconnected partition of $G_{ab}$. Otherwise, let $b \in B \cap V_2$ and let $a$ be an arbitrary vertex of $A$ that is adjacent to $b$ (which exists by the definition of a proper W-join). Then $V_1 \setminus ((A \setminus\{a\}) \cup B),V_2 \setminus (B \setminus\{b\}),V_3,V_4$ is a disconnected partition of $G_{ab}$. 

Now assume that $A \subseteq V_2$. Note that $B \subseteq V_1 \cup V_2 \cup V_3$. Since $B$ is a clique and $V_1$ is anticomplete to $V_3$, it follows that $B \subseteq V_1 \cup V_2$ or $B \subseteq V_2 \cup V_3$. First, assume that $B \subseteq V_2$. Let $a \in A$ and $b \in B$ be arbitrary adjacent vertices; note that $a,b \in V_2$. Then $V_1,V_2 \setminus ((A \setminus\{a\}) \cup (B \setminus\{b\})),V_3,V_4$ is a disconnected partition of $G_{ab}$. So we may assume that $B \not\subseteq V_2$. Then $B \cap V_1 \not= \emptyset$ or $B \cap V_3 \not= \emptyset$. Without loss of generality, we assume is the former. Let $b \in B \cap V_1$ and let $a \in A$ be any neighbour of $b$. Note that $a \in V_2$. Then $V_1 \setminus (B \setminus\{b\}),V_2 \setminus ((A \setminus\{a\}) \cup B),
V_3,V_4$ is a disconnected partition of $G_{ab}$.

It remains to consider the case where $A \cap V_1 \not= \emptyset$ and $A \cap V_2 \not= \emptyset$. Let $P = N(A) \setminus N[B]$, $Q = N(B) \setminus N[A]$, $M = N[A \cup B] \setminus (P \cup Q)$, and $R = V(G) \setminus (P \cup Q \cup M)$.
Note that $P$ is complete to $A$ and anticomplete to $B$, whereas $Q$ is complete to $B$ and anticomplete to~$A$. Moreover, $M$ is complete to $A\cup B$, whereas $R$ is anticomplete to $A\cup B$.
Then, by the assumptions of the case, we have that $P \subseteq V_1 \cup V_2$. Note that $B \subseteq V_1 \cup V_2 \cup V_3$. Since $B$ is a clique and $V_1$ is anticomplete to $V_3$, it follows that $B \subseteq V_1 \cup V_2$ or $B \subseteq V_2 \cup V_3$.
Moreover, as $A\cap V_1\not= \emptyset$, it follows from the definition of a proper W-join that $B \not\subseteq V_3$.
We now prove that $B \subseteq V_1 \cup V_2$.

For contradiction, assume 
$B \cap V_3 \not= \emptyset$ and thus $B \cap V_2 \not= \emptyset$. 
As $M$ is complete to $A$ and~$B$ and $A \cup B$ has a nonempty intersection with each of $V_1,V_2,V_3$, it follows from the definition of a disconnected partition that $M \subseteq V_2$. Similarly, we derive that
$Q \subseteq V_2 \cup V_3$; recall also that $P \subseteq V_1 \cup V_2$. Suppose $V_1 \setminus A \not= \emptyset$. Then $V_1 \setminus A, V_2 \cup A, V_3, V_4$ is also a disconnected partition of $G$, contradicting our choice of the disconnected partition $V_1,V_2,V_3,V_4$. Hence, $V_1 \setminus A = \emptyset$ and thus, $V_1 \subseteq A$. Then $P \subseteq V_2$. By the definition of a W-join, any path of length~$2$ from a vertex in $R$ to a vertex in $A$ must intersect $P$ or~$M$. As $M \cup P \subseteq V_2$ and $V_4$ is anticomplete to $V_2$, we obtain $R \cap V_4 = \emptyset$. Since $A \cup B \cup P \cup M \cup Q \cup R = V(G)$ and none of $A,B,P,M,Q,R$ intersects $V_4$, it follows that $V_4 = \emptyset$, a contradiction.

We may thus assume that $B \subseteq V_1 \cup V_2$. First, assume that there exist adjacent vertices $a \in A$ and $b \in B$ such that $|V_1 \cap \{a,b\}| = 1$ (and thus $|V_2 \cap \{a,b\}| = 1$). Then $V_1 \setminus ((A \cup B)\setminus\{a,b\}),V_2 \setminus ((A \cup B)\setminus\{a,b\})$ is a disconnected partition of $G_{ab}$. Hence, we may assume that no such two vertices exist. It follows that neither $B \subseteq V_1$ nor $B \subseteq V_2$; otherwise, such $a$ and $b$ would exist by the definition of a proper W-join. Then we may conclude that $(A,B)$ is partitionable with sets $A \cap V_1, A \cap V_2, B \cap V_1, B \cap V_2$. Since $(A,B)$ is unshatterable, it follows from Lemma~\ref{l-shatterc4} that $G[A \cup B]$ is isomorphic to $C_4$. 
Hence $|A|=|B|=2$ and $V_1,V_2$ each contain exactly one vertex of $A$ and exactly one vertex of $B$.
Since $G$ is not cobipartite and 
$G$ is connected (as $G$ has diameter~2), it follows that one of $P,M,Q$ is non-empty. However, each vertex in $P \cup M \cup Q$ is adjacent to a vertex of $V_1$ and a vertex of $V_2$. Hence, $P \cup M \cup Q \subseteq V_1 \cup V_2$. Without loss of generality, $(P \cup M \cup Q) \cap V_1 \not= \emptyset$. Let $a$ be the single vertex of $A \cap V_2$ and let $b$ be the single vertex of $B \cap V_2$. Then $V_1 \setminus (A \cup B),V_2,V_3,V_4$ is a disconnected partition of $G_{ab}$. The lemma follows.
\qed
\end{proof}

In Section~\ref{s-2} we will show that by iterating the above lemma, we can remove all W-joins from an input claw-free graph of diameter~2. However, to this end, it is crucial to have
a polynomial-time algorithm that actually finds an unshatterable proper W-join (if it exists). 
Our algorithm for this problem relies on the $O(n^2m)$-time algorithm by King and Reed~\cite{KingR2008} to find a proper W-join $(A,B)$. We test in linear time whether the proper W-join is partitionable by considering the graph $H$ obtained from $G[A \cup B]$ by removing all edges with both endpoints in $A$ or in $B$. We argue that we can recurse on a smaller proper W-join if $H$ has two or more connected components, and that $(A,B)$ is unshatterable otherwise.

\begin{lemma}\label{l-shatterfind}
Let $G$ be a graph with distinct neighbourhoods. Then in $O(n^2m)$ time, we can find an unshatterable proper W-join in $G$, or report that $G$ has no proper W-join.
\end{lemma}

\begin{proof}
King and Reed~\cite{KingR2008} proved that in $O(n^2m)$ time, one can find a proper W-join in $G$, or report that $G$ does not admit a proper W-join. In the former case, let $(A,B)$ be the proper W-join that is found. If $G[A \cup B]$ is isomorphic to $C_4$, which can be checked in constant time, then $(A,B)$ is unshatterable, and we return $(A,B)$. So assume otherwise. We can test in linear time whether $(A,B)$ is partitionable as follows. Let $H$ be the graph obtained from $G[A \cup B]$ by removing all edges with both endpoints in $A$ or in $B$. Note that $H$ is bipartite. Observe also that $H$ cannot have any singleton connected components, because such a vertex (say $a \in A$) would be anticomplete to the other side ($B$ in this case), a contradiction to the definition of a proper W-join. Now note that $(A,B)$ is partitionable if and only if $H$ has two or more connected components. If $H$ has one connected component, which can be tested in linear time, then $(A,B)$ is not partitionable and thus unshatterable. If $H$ has three or more connected components, then let $C$ be any such connected component. Then $(A,B)$ is partitionable with $V(C) \cap A, V(C)\cap B, (V(H) \setminus V(C)) \cap A, (V(H) \setminus V(C)) \cap B$, and since $H$ does not have any singleton components, $|V(H) \setminus V(C)| > 3$, and $G$ has distinct neighbourhoods, $(V(H) \setminus V(C)) \cap A, (V(H) \setminus V(C)) \cap B$ is a W-join in $G$. It follows from Lemma~\ref{l-mustbeproper} that it must be a proper W-join. Now repeat the algorithm for $(V(H) \setminus V(C)) \cap A, (V(H) \setminus V(C)) \cap B$.  If $H$ has exactly two connected components, then let $C$ be a connected component of $H$ with the most vertices. Since $G[A \cup B]$ is not isomorphic to $C_4$, it follows that $C$ has at least three vertices. Then, using that $G$ has distinct neighbourhoods, $(V(C) \cap A, V(C)\cap B)$ is a W-join in $G$, which by Lemma~\ref{l-mustbeproper} must be a proper W-join. Now repeat the algorithm for $(V(C) \cap A, V(C)\cap B)$. This gives the required algorithm.

Observe that the algorithm can recurse at most $n$ times, since in each recursion step it considers a strictly smaller proper W-join. Each recursion step takes linear time by performing a breadth-first search on the graph $H$. Hence, the running time of our algorithm is dominated by the initial call to the algorithm, which takes $O(n^2m)$ time.
\qed
\end{proof}

\subsection{Structure of Claw-Free Graphs and Solving Disconnected Cut}\label{s-2}

We now show a decomposition of claw-free graphs of diameter~2.
In order to dos we need a number of definitions, which all originate in Chudnovsky and Seymour~\cite{CS05}, 
but are reformulated in the style of Hermelin et al.~\cite{HermelinML14,HMLW11-arxiv,HMLW11}. A \emph{trigraph} is defined by a set of vertices and an adjacency relation where any two vertices are either strongly adjacent, semi-adjacent, or strongly anti-adjacent, and every vertex is semi-adjacent to at most one vertex. One may think of a trigraph as a normal graph where some edges are simultaneously present and non-present. In particular, a trigraph without semi-adjacent pairs of vertices is just a graph. We call vertices $u,v$ of a trigraph \emph{adjacent} if they are strongly adjacent or semi-adjacent, and \emph{anti-adjacent} if they are strongly anti-adjacent or semi-adjacent. We call two sets $X,Y$ \emph{(strongly) complete} if each pair of vertices $v \in X, w \in Y$ is (strongly) adjacent, and \emph{(strongly) anti-complete} if each pair of vertices $v \in X, w \in Y$ is (strongly) anti-adjacent.

A graph $H$ is a \emph{thickening} of a trigraph $G$ if there is a partition of $V(H)$ into non-empty sets $X_v$ for each $v \in V(G)$ such that:
\begin{itemize}
\item $X_v$ is a clique for each $v \in V(G)$;
\item if $v,w$ are strongly adjacent in $G$, then $X_v$ is complete to $X_w$;
\item if $v,w$ are strongly anti-adjacent in $G$, then $X_v$ is anti-complete to $X_w$;
\item if $v,w$ are semi-adjacent in $G$, then $X_v$ is neither complete nor anti-complete to $X_w$.
\end{itemize}
A pair of vertices $v,w$ in a graph $H$ form \emph{twins} if $N[v] = N[w]$.

Let $H$ be a thickening of a trigraph $G$. If $v \in V(G)$ is not semi-adjacent to another vertex, then the vertices in $X_v$ form twins. Furthermore, if $v,w$ are semi-adjacent, then $(X_v,X_w)$ is a W-join in $G$ (recall that $v$ and $w$ are not semi-adjacent to any other vertices in $G$). In particular, if $H$ contains neither twins nor W-joins, 
then $G$ and $H$ are isomorphic.

A {\it strip-structure} of a connected graph $G$ consists of a connected multigraph $H$ (with parallel edges and self-loops), a nonempty set $X_e \subseteq V(G)$ for each $e \in E(H)$, and a nonempty set $X_{e,y} \subseteq X_e$ for each $e \in E(H)$ and $y \in V(H)$ such that $e$ is incident to $y$, such that
\begin{itemize}
\item the sets $X_e$ partition $V(G)$;
\item for each $e \in E(H)$ incident with two vertices $y,y' \in V(H)$, each vertex in $X_{e,y} \cap X_{e,y'}$ is anti-complete to $X_{e} \setminus (X_{e,y} \cup X_{e,y'})$;
\item for each $y \in V(H)$, the graph induced by the union, over all $e \in E(H)$ incident to $y$, of the sets $X_{e,y}$ is a clique;
\item if $v,w$ are adjacent in $G$, then either $v,w \in X_e$ for some $e \in E(H)$ or there exist $e,e' \in E(H)$ incident with the same vertex $y \in V(H)$ for which $v \in X_{e,y}$ and $w \in X_{e',y}$.
\end{itemize}
For each $e \in E(H)$, the \emph{strip corresponding to $e$} is a pair $(J,Z)$, where $Z$ is a set of new vertices, one for each vertex $y \in V(H)$ incident with $e$, and the graph $J$ is obtained from $G[X_e]$ by adding $Z$ and for each $z \in Z$, making $z$ complete to $X_{e,y}$, where $y \in V(H)$ is the vertex corresponding to $z$. The definition of a strip-structure implies that each 
strip~$(J,Z)$ has
$|Z|=1$ (if $e$ is a self-loop) or $|Z|=2$ (otherwise).  
We may think of $Z$ as being `representatives' of the rest of the graph, but note that the vertices of $Z$ are not part of $G$.

\tikzset{vertex/.style={minimum size=2mm,circle,fill=black,draw,inner sep=0pt},
        decoration={markings,mark=at position .5 with
{\arrow[black,thick]{stealth}}}}

\begin{figure}[t]
 \centering
 \begin{tikzpicture}[scale=0.5]
   \draw[fill=green,fill opacity=0.2] (0,-4.5) ellipse (3.2cm and 1.4cm);
   \draw[fill=green] (2,-4.5) ellipse (0.5cm and 1cm);
   \draw[fill=green] (-2,-4.5) ellipse (0.5cm and 1cm);
   \draw[fill=red, fill opacity=0.2] (-3,-2) ellipse (0.6cm and 0.6cm);
   \draw[fill=red] (-3,-2) ellipse (0.35cm and 0.35cm);
   \draw[fill=orange, fill opacity=0.2] (0,-2) ellipse (1.8cm and 0.6cm);
   \draw[fill=orange] (-1,-2) ellipse (0.35cm and 0.35cm);
   \draw[fill=orange] (1,-2) ellipse (0.35cm and 0.35cm);
   \draw[fill=purple, fill opacity=0.2] (3,-2) ellipse (0.6cm and 0.6cm);
   \draw[fill=purple] (3,-2) ellipse (0.35cm and 0.35cm);
   \draw[fill=blue, fill opacity=0.2] (-3,0) ellipse (0.6cm and 0.6cm);
   \draw[fill=blue] (-3,0) ellipse (0.35cm and 0.35cm);
   \draw[fill=gray,fill opacity=0.2] (-0.5,1.5) ellipse (1.2cm and 1.2cm);
   \draw[fill=gray] (-1,1) ellipse (0.45cm and 0.45cm);
   \draw[fill=gray] (0,1) ellipse (0.45cm and 0.45cm);
  \node (a) at (-3.95,-2){$a$};
  \node (b) at (-3.95,0){$b$};
  \node (c) at (-2.3,1.75){$c$};
  \node (d) at (0,-1){$d$};
  \node (g) at (3.95,-2){$f$};
  \node (h) at (0,-6.5){$g$};
   \node (6) at (-2,-4) [vertex]{};
   \node (7) at (2,-4) [vertex]{};
   \node (8) at (-2,-5) [vertex]{};
   \node (9) at (2,-5) [vertex]{};
   \draw[thick] (7)--(9);
   \draw[thick] (6)--(8);
   \node (10) at (-1,-4) [vertex]{};
   \node (11) at (-0.5,-5) [vertex]{};
   \node (12) at (0.5,-4) [vertex]{};
   \node (12a) at (0.3,-5) [vertex]{};
   \draw[thick] (6)--(10);
   \draw[thick] (6)--(11);
   \draw[thick] (8)--(11);
   \draw[thick] (11)--(12);
   \draw[thick] (10)--(11);
   \draw[thick] (10)--(12);
   \draw[thick] (12)--(7);
   \draw[thick] (12)--(9);
   \draw[thick] (12a)--(10);
   \draw[thick] (12a)--(12);
   \draw[thick] (12a) -- (7);
   \draw[thick] (12a) -- (9);
   \node (13) at (3,-2) [vertex]{};
   \node (14) at (-3,-2) [vertex]{};
   \draw[thick] (13)--(7);
   \draw[thick] (13)--(9);
   \draw[thick] (14)--(6);
   \draw[thick] (14)--(8);
   \node (15a) at (-1,-2) [vertex]{};
   \draw[thick] (15a)--(13);
   \draw[thick] (15a)--(14);
   \node (15b) at (1,-2) [vertex]{};
   \node (15c) at (0,-2) [vertex]{};
   \node (16) at (-1,1) [vertex]{};
   \node (17) at (0,1) [vertex]{};
   \draw[thick] (16)--(17);
   \draw[thick] (17)--(13);
   \draw[thick] (17)--(15b);
   \draw[thick] (16)--(15a);
   \draw[thick] (16)--(14);
   \node (19) at (-3,0) [vertex]{};
   \draw[thick] (19)--(14);
   \draw[thick] (19)--(16);
   \draw[thick] (13)--(14);
   \node (20) at (-1,2) [vertex]{};
   \node (21) at (0,2) [vertex]{};
   \node (21a) at (-0.5,2) [vertex]{};
   \draw[thick] (21)--(17);
   \draw[thick] (20)--(16);
   \draw[thick] (20)--(21);
   \draw[thick] (15a)--(19);
 \end{tikzpicture}
 \tikzset{every loop/.style={min distance=20mm,looseness=10}}
 \begin{tikzpicture}
  \node (0) at (-3,2) [vertex]{};
  \node (1) at (0,2) [vertex]{};
  \node (2) at (-3,3.5) [vertex]{};
  \node (3) at (0,3.5) [vertex]{};
  \draw[thick,green] (0)--(1);
  \draw[thick,purple] (1)--(3);
  \draw[thick,red] (2)--(0);
  \draw[thick,orange] (2)--(3);
  \node (hh) at (-4,4.5){};
  \draw[thick,blue] (2) -- (hh);
  \draw[thick,gray] (2) to[bend right=-50] (3);
  \node[red] (a) at (-3.25,2.75){$a$};
  \node[blue] (b) at (-3.5,3.75){$b$};
  \node[gray] (c) at (-1.5,4.45){$c$};
  \node[orange] (d) at (-1.5,3.25){$d$};
  \node[purple] (g) at (0.25,2.75){$f$};
  \node[green] (h) at (-1.5,1.75){$g$};
  \node at (-3,1){};
 \end{tikzpicture}
 \qquad
 \begin{tikzpicture}[scale=0.5]
   \draw[fill=red, fill opacity=0.2] (0,0) ellipse (1cm and 1cm);
   \draw[fill=red] (0,0) ellipse (0.5cm and 0.5cm);
   \node (1) at (0,0) [vertex] {};
   \node (2) at (-2,0) [vertex] {};
   \node (3) at (2,0) [vertex] {};
   \draw[thick] (1)--(2);
   \draw[thick] (1)--(3);
   \node at (0,-1.5){trivial line graph strip};
   \node at (-2,-0.8){$Z$};
   \node at (2,-0.8){$Z$};
   \draw[fill=green, fill opacity=0.2] (0,-4.5) ellipse (2.5cm and 2cm);
   \draw[fill=green] (2,-4.5) ellipse (0.5cm and 1cm);
   \draw[fill=green] (-2,-4.5) ellipse (0.5cm and 1cm);
   \node (4) at (-3,-4.5) [vertex]{};
   \node (5) at (3,-4.5) [vertex]{};
   \node (6) at (-2,-4) [vertex]{};
   \node (7) at (2,-4) [vertex]{};
   \node (8) at (-2,-5) [vertex]{};
   \node (9) at (2,-5) [vertex]{};
   \draw[thick] (7)--(9);
   \draw[thick] (6)--(8);
   \draw[thick] (4)--(6);
   \draw[thick] (4)--(8);
   \draw[thick] (5)--(7);
   \draw[thick] (5)--(9);
   \node (10) at (-1,-4) [vertex]{};
   \node (11) at (-0.5,-5) [vertex]{};
   \node (12) at (0.5,-4) [vertex]{};
   \node (12a) at (0.3,-5) [vertex]{};
   \draw[thick] (6)--(10);
   \draw[thick] (6)--(11);
   \draw[thick] (8)--(11);
   \draw[thick] (11)--(12);
   \draw[thick] (10)--(11);
   \draw[thick] (10)--(12);
   \draw[thick] (12)--(7);
   \draw[thick] (12)--(9);
   \draw[thick] (12a)--(10);
   \draw[thick] (12a)--(12);
   \draw[thick] (12a) -- (7);
   \draw[thick] (12a) -- (9);
   \node at (0,-7){stripe};
   \node at (-3.2,-5.3){$Z$};
   \node at (3.0,-5.3){$Z$};
 \end{tikzpicture}
 \quad
 \quad
\caption{This figure is inspired by~\cite[Fig.~1]{HermelinML14}. The left figure is a claw-free graph $G$. A strip-structure is indicated by the coloured ovals. The middle figure shows the graph $H$ underlying the strip-structure; the edge $b$ indicates a self-loop. The lighter-coloured ovals in the left figure show the sets $X_e$ for $e \in E(H)$. The darker-coloured ovals show the sets $X_{e,y}$ for each $e \in E(H)$ and each vertex $y \in V(H)$ to which $e$ is incident. 
The right figure shows an example of the trivial line graph strip with $|Z|=1$ above a stripe with $|Z|=2$. The lighter-coloured ovals show $X_e = V(J) \setminus Z$ and the darker-coloured ovals show $X_{e,y}$ for each $y \in V(H)$ to which $e$ is incident. Note that strips $a$, $b$, and $f$ in the left figure are trivial line graph strips. Strips $c$, $d$, and $g$ are stripes and might look different depending on $X_e$; the stripe in the right panel corresponds to $f$ (the stripes corresponding to $c$ and $d$ are not pictured). Observe that indeed the set $Z \subseteq V(J)$ is not part of $G$.}
\label{fig:stripstructure}
\end{figure}
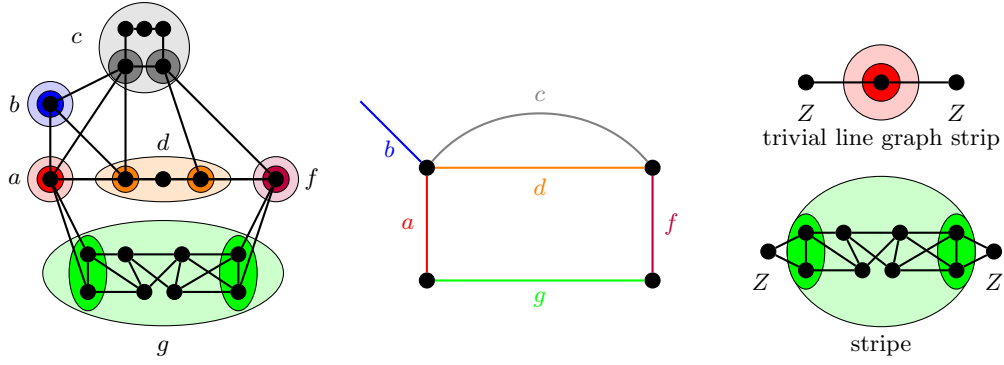

A strip $(J,Z)$ is a \emph{trivial line graph strip} if $|Z|=1$ and $J$ is a 2-vertex path, or if $|Z|=2$ and $J$ is a 3-vertex path. A strip $(J,Z)$ is a \emph{stripe} if no vertex in $V(J) \setminus Z$ is adjacent to more than one vertex of $Z$. 
In particular $(J,Z)$ is a stripe if $|Z|=1$, and moreover, if $(J,Z)$ is a stripe with $|Z|=2$ for $e=yy'$, then $X_{e,y}$ and $X_{e,y'}$ are disjoint.
A \emph{thickening of a stripe} $(J,Z)$ is defined as usual, except that it must hold that $|X_z|=1$ for each $z \in Z$.

We now define an XX-trigraph, XX-graphs, and XX-trigraph stripes. The details are actually unimportant, because we will only use the fact that XX-(tri)graphs (stripes) have at most $13$~vertices.
Let $G$ be a trigraph on vertices $v_1,\ldots,v_{13}$ such that $v_i$ is strongly adjacent to $v_{i+1}$ for $i=1,2,3,4,5$; $v_6$ is strongly adjacent to $v_1$; $v_i$ is anti-adjacent to $v_j$ for $1 \leq i < j \leq 6$ with $1 < j-i < 5$; $v_{7}$ is strongly adjacent to $v_{1}$ and $v_{2}$; $v_{8}$ is strongly adjacent to $v_{4}$, $v_{5}$, and possibly adjacent to $v_{7}$; $v_{9}$ is strongly adjacent to $v_{1}$, $v_{2}$, $v_{3}$, and $v_{6}$; $v_{10}$ is strongly adjacent to $v_{3}$, $v_{4}$, $v_{5}$, and $v_{6}$, and adjacent to $v_{9}$; $v_{11}$ is strongly adjacent to $v_{1}$, $v_{3}$, $v_{4}$, $v_{6}$, $v_{9}$, and $v_{10}$; $v_{12}$ is strongly adjacent to $v_{2}$, $v_{3}$, $v_{5}$, $v_{6}$, $v_{9}$, and $v_{10}$; $v_{13}$ is strongly adjacent to $v_{1}$, $v_{2}$, $v_{4}$, $v_{5}$, $v_{7}$, and $v_{8}$. Then we call $G - X$ for any $X \subseteq \{v_7,v_{11},v_{12},v_{13}\}$ an \emph{XX-trigraph}. We call $G-X$ an \emph{XX-graph} if it has no semi-adjacent vertices. We call a stripe $(J,Z)$ of a strip-structure an XX-(tri)graph stripe if $J$ is an XX-(tri)graph with vertex set $\{v_1,\ldots,v_{13}\} \setminus X$ for some $X \subseteq \{v_7,v_{11},v_{12},v_{13}\}$ such that $v_7$ and $v_8$ are strongly anti-adjacent and $Z = \{v_7,v_8\} \setminus X$. Note that for any (thickening of an) XX-(tri)graph stripe, it holds that $V(J) \setminus N[Z] \not= \emptyset$ (take $v_6$ or $X_{v_6}$).

\begin{theorem}[{\cite[Theorem 6.8, simplified]{HMLW11-arxiv}; see also \cite{HMLW11}}] \label{thm:orig-claw-decomposition}
Let $G$ be a connected claw-free graph with $\alpha(G) > 3$ such that $G$ does not admit twins or proper W-joins. Then
\begin{itemize}
\item $G$ is a thickening of an XX-trigraph,
\item $G$ is a proper circular-arc graph, or
\item $G$ admits a strip-structure such that for each strip $(J,Z)$
\begin{itemize}
\item $(J,Z)$ is a trivial line graph strip, or
\item $(J,Z)$ is a stripe for which $J$ is connected and
\begin{itemize}
\item $\alpha(J) \leq 3$, and $V(J) \setminus N_J[Z] \not= \emptyset$,
\item $|Z|=1$ and $J$ is a proper circular-arc graph,
\item $|Z|=2$ and $J$ is a proper interval graph, or
\item $(J,Z)$ is a thickening of an XX-trigraph stripe.
\end{itemize}
\end{itemize}
\end{itemize}
Moreover, we can distinguish the cases and find the strip-structure in polynomial time.
\end{theorem}

\begin{theorem} \label{thm:claw-decomposition}
Let $G$ be a connected claw-free graph with $\alpha(G) > 3$ such that $G$ does not admit twins or W-joins. Then
\begin{itemize}
\item $G$ is an XX-graph,
\item $G$ is a proper circular-arc graph, or
\item $G$ admits a strip-structure such that for each strip $(J,Z)$
\begin{itemize}
\item $(J,Z)$ is a trivial line graph strip, or
\item $(J,Z)$ is a stripe for which $J$ is connected and
\begin{itemize}
\item $\alpha(J) \leq 3$, and $V(J) \setminus N_J[Z] \not= \emptyset$,
\item $|Z|=1$ and $J$ is a proper circular-arc graph,
\item $|Z|=2$ and $J$ is a proper interval graph, or
\item $(J,Z)$ is an XX-graph stripe.
\end{itemize}
\end{itemize}
\end{itemize}
Moreover, we can distinguish the cases and find the strip-structure in polynomial time.
\end{theorem}
\begin{proof}
Observe that if $G$ does not admit twins nor W-joins, then thickening operations effectively do nothing, as discussed previously. Hence, any thickening of an XX-trigraph (stripe) is in fact an XX-graph (stripe).
\qed
\end{proof}

We are now ready to prove Theorem~\ref{thm:claw}.

\begin{theorem}\label{thm:claw}
Every claw-free graph $G$ of diameter~$2$ with distinct neighbourhoods, no W-joins, $\alpha(G) > 3$, and $|V(G)| > 13$ is a proper circular-arc graph or a line graph.
\end{theorem}

\begin{proof}
Let $G$ be as in the theorem statement, and assume that $G$ is not a proper circular-arc graph. In order to prove that $G$ is indeed a line graph, we apply Theorem~\ref{thm:claw-decomposition} to the graph $G$. By the assumptions on $G$, we obtain a strip-structure consisting of strips of several possible types. We will argue that this strip-structure can be modified so that it contains only trivial line graph strips. This implies that $G$ is a line graph.

We first claim that if the strip-structure contains a stripe $(J,Z)$ such that $V(J) \setminus N_J[Z] \not= \emptyset$ and $1 \leq |Z| \leq 2$, then either $|Z| = 1$ and $G$ is isomorphic to $J$ or to $J-Z$, or $|Z|=2$ and $G$ is isomorphic to $J-Z'$ for some $Z' \subseteq Z$ or to $J$ where the vertices in $Z$ have been identified or made adjacent.

Let $(J,Z)$ be a stripe such that $V(J) \setminus N_J[Z] \not= \emptyset$ and $1 \leq |Z| \leq 2$. Let $x \in V(J) \setminus N_J[Z]$. Since $G$ has diameter~$2$, every vertex of $V(G) \setminus (V(J) \setminus Z)$ must be in $N_G(N_J(Z))$. Now observe that $N_G(N_J(z)) \setminus V(J)$ is a clique for each $z \in Z$ by the definition of a strip-structure.

Consider the case $|Z| = 1$. Let $X = N_G(N_J(Z)) \setminus V(J)$. Note that the neighbourhood of every vertex in $X$ is $N_J(Z)$. Since $X$ is a clique and $G$ has distinct neighbourhoods, it follows that $|X| \leq 1$. The claim follows.

Consider the case $|Z| = 2$ and let $Z = \{z_1,z_2\}$. Let $X = N_G(N_J(z_1)) \setminus V(J)$ and $Y = N_G(N_J(z_2)) \setminus V(J))$. Suppose that $|X \cup Y| = |X \cap Y|$. Then all vertices in $X \cup Y$ have the same neighbourhood, namely $X \cup Y \cup N_J(Z)$, which by the distinct neighbourhood assumption implies that $|X \cup Y|=1$. This implies the claim. A similar line of reasoning applies if $|X \cap Y| = 0$ and $X$ and $Y$ are anticomplete.
Otherwise, suppose that $|X \cup Y| > 2$. Then let $A = X \setminus Y$, $B = Y \setminus X$, and iteratively assign the vertices of $X \cup Y$ to the smallest of $A$ and $B$. Then $(A,B)$ is a W-join in $G$, a contradiction. Hence, we may assume that $|X \cup Y| \leq 2$. Suppose that $|X \cup Y| = |X \cap Y|=2$. Then both vertices in $X \cup Y$ have the same neighbourhood, namely $X \cup Y \cup N_J(Z)$, a contradiction. Then $|X \cap Y|=1$, which implies the claim, or otherwise, $|X \cup Y| \leq 2$, which then also implies the claim.

Now consider the possible stripes in the strip-structure. If $1 \leq |Z| \leq 2$, $\alpha(J) \leq 3$, and $V(J)\setminus N[Z] \not= \emptyset$, then applying the claim, it follows that $\alpha(G) \leq 3$, a contradiction. If $(J,Z)$ is an XX-graph stripe, then $V(J) \setminus N[Z] \not= \emptyset$, and thus applying the claim, $G$ has at most $13$ vertices, a contradiction.

Suppose that $|Z| = 1$ and $J$ is a proper circular-arc graph. If $V(J)\setminus N[Z] = \emptyset$, then the stripe can be decomposed into $|V(J)| - 1$ trivial line graph strips, one for each vertex in $V(J) \setminus Z$. 
Otherwise, that is, if $V(J)\setminus N[Z] \not= \emptyset$, then applying the claim, it follows that $G$ is a proper circular-arc graph, a contradiction.

Suppose that $|Z| = 2$ and $J$ is a proper interval graph. Let $Z=\{z_1,z_2\}$. If $V(J)\setminus N[Z] = \emptyset$, then $(N_J(z_1),N_J(z_2))$ would induce a W-join in $G$, a contradiction unless $|N_J(z_1)|=|N_J(z_2)|=1$. That implies $|V(J)| = 4$ and $J$ is in fact a four-vertex path. Then the stripe can be decomposed into two trivial line graph strips by adding a new node to the strip-structure that supports them both. Otherwise, that is if $V(J)\setminus N[Z] \not= \emptyset$, then applying the claim, $G$ is a proper interval graph or a proper circular-arc graph. The latter can be seen from the fact that there trivially exists a representation of $J$ as a proper interval graph in which the intervals corresponding to $Z$ extend farthest left and right in the representation. By bending the representation around the circle, we obtain a proper circular-arc graph. In either case, we obtain a contradiction.

From this, it follows that each strip in the (modified) strip-structure must be a trivial line graph strip. This implies that $G$ is a line graph, as claimed.
\qed
\end{proof}

We now apply Theorem~\ref{thm:claw} and results from previous sections to obtain our main result.

\begin{theorem}\label{t-mainmain}
{\sc Disconnected Cut} is $O(n^3m)$-time solvable for claw-free graphs.
\end{theorem}

\begin{proof}
Let $G$ be a connected claw-free graph on $n$ vertices and $m$ edges. Assume $n\geq 14$.
We compute the diameter of $G$ in $O(n^2)$ time.
 By Lemma~\ref{l-diameter}, $G$ has no disconnected cut if its diameter is~$1$ and has a disconnected cut if its diameter is at least~$3$. Assume the diameter of $G$ is~$2$.
We check if $\alpha(G) \leq 3$ in $O(n(m+n \log n))$ time~\cite{FOS14}.
If so, then we decide if $G$ has a disconnected cut in $O(n^3)$ time by Lemma~\ref{l-4p1}. 
Assume $\alpha(G) > 3$. Hence, $G$ is not cobipartite.

Next, we check whether $G$ contains a vertex $u$ for which there exists a vertex $v$ such that $N(u) \setminus\{v\} \subseteq N(v) \setminus\{u\}$. This takes $O(n^3)$ time. If so, then we remove $u$ from $G$ (and restart the algorithm with the resulting graph, which is still connected and claw-free). 
This is correct by Lemma~\ref{l-neighbourhood}. Hence, we may assume that $G$ has distinct neighbourhoods.

Then, we get rid of all W-joins in $G$. Since $G$ has distinct neighbourhoods, it follows from Lemma~\ref{l-mustbeproper} that every W-join in $G$ is a proper W-join. Using Lemma~\ref{l-shatterfind}, in $O(n^2m)$ time, we can find an unshatterable W-join in $G$ or correctly decide that $G$ does not admit a proper W-join (and thus no W-join). In the former case, we apply Lemma~\ref{l-noproper} on the unshatterable proper W-join $(A,B)$ that is found. This takes linear time. We then restart the algorithm on the graph $G_{ab}$ found by Lemma~\ref{l-noproper} (note that $G_{ab}$ is still connected and claw-free).
Since $|A|+|B| \geq 3$, $|V(G_{ab})| < |V(G)|$ and thus we can recurse at most $n$ times. Hence, we may assume that $G$ admits no W-joins.

Next, we check if $G$ is a circular-arc graph in linear time by Lemma~\ref{l-mcc}. 
If so, then we apply Theorem~\ref{t-circulararc} to decide if $G$ has a disconnected cut in 
$O(n^2)$ time. Hence, we may assume that $G$ is not (proper) circular-arc.
By Theorem~\ref{thm:claw} this means that $G$ is a line graph. Hence, we apply Theorem~\ref{t-linegraphs} to decide whether $G$ admits a disconnected cut in $O(n^4)$ time. This finishes the description of the algorithm. The running time is clearly $O(n^3m)$.
\qed
\end{proof}

Recall from~\cite{IKPT11} and~\cite{FMPS09} that {\sc $C_4$-Contractibility} and {\sc ${\cal C}_4$-Compaction}, respectively, are equivalent to {\sc Disconnected Cut} on graphs of diameter~$2$.
Combining these claims with Theorem~\ref{thm:claw} leads to the following two consequences.

\begin{corollary}\label{c-con}
{\sc $C_4$-Contractibility} is $O(n^3m)$-time solvable for claw-free graphs of diameter~$2$.
\end{corollary}

\begin{corollary}\label{c-com}
{\sc ${\cal C}_4$-Compaction} is $O(n^3m)$-time solvable for claw-free graphs of diameter~$2$.
\end{corollary}

\section{${\mathbf H}$-Free Graphs for Graphs ${\mathbf H}$ with at Most Four Vertices}\label{s-k4}

In this section we show that the $K_4$ is the only 4-vertex graph~$H$ for which the computational complexity of {\sc Disconnected Cut} is still open on $H$-free graphs. To prove this claim, we need to show three additional results.

A graph~$G=(V,E)$ is \emph{complete $k$-partite} if~$V$ can be partitioned into~$k$ independent sets $A_1,\ldots,A_k$ for some integer~$k\geq 2$, such that two vertices are adjacent if and only if they belong to two different sets~$A_i$ and~$A_j$. The graph $\overline{P_1+ P_3}$, which is a triangle with a pendant vertex, is also known as the {\it paw}. Olariu proved the following result for paw-free graphs. 

\begin{lemma}[\cite{Ol88}]\label{l-paw}
Every connected $\overline{P_1+P_3}$-free graph is either $C_3$-free or complete $k$-partite for some $k\geq 3$.
\end{lemma}

We also need the following lemma from~\cite{FMPS09} for $C_3$-free graphs (or we could use the more general result that 
{\sc Disconnected Cut} is polynomial-time solvable on $(C_3+P_1)$-free graphs in~\cite{DMS12}).

\begin{lemma}[\cite{FMPS09}]\label{l-c3}
{\sc Disconnected Cut} is $O(n^3)$-time solvable for $C_3$-free graphs.
\end{lemma}

\begin{lemma}\label{l-paw2}
{\sc Disconnected Cut} is $O(n^3)$-time solvable for $\overline{P_1+P_3}$-free graphs.
\end{lemma}

\begin{proof}
Let $G$ be a $\overline{P_1+P_3}$-free graph. By using brute force we check in $O(n^3)$ time if $G$ contains a $C_3$. 
If $G$ is $C_3$-free, then we apply Lemma~\ref{l-c3}.
Otherwise, by Lemma~\ref{l-paw}, $G$ is complete $k$-partite for some integer $k\geq 3$.
We claim that $G$ has no disconnected cut. For contradiction, let $V_1,V_2,V_3,V_4$ be a disconnected partition of $G$.
Let $A_1,\ldots,A_k$ be the partition classes of $G$.
We may assume without loss of generality that $V_1 \cap A_1 \not= \emptyset$. Since $A_1$ is complete to $A_i$ for each $i \not= 1$ and $V_1$ is anticomplete to $V_3$, we know that $V_3 \cap A_i = \emptyset$ for all $i \not= 1$. It follows that $\emptyset \subset V_3 \subseteq A_1$. By the same argument, $\emptyset \subset V_1 \subseteq A_1$. 
The same line of reasoning implies that $\emptyset \subset V_i \subseteq A_2$ and $\emptyset \subset V_i \subseteq A_4$ for some $i$, say $i\in \{1,2\}$. Then $A_3 \cap (V_1 \cup V_2 \cup V_3 \cup V_4) = \emptyset$, a contradiction.
\qed
\end{proof}

We recall that a biclique is a complete bipartite graph $K_{r,s}$ for some integers $r,s\geq 1$. 
Recall also that the $2K_2$-{\sc Partition} problem is to decide if a graph $G$ has a $2K_2$-partition, or equivalently, if
the vertex set of~$G$ can be partitioned into two 
non-empty sets $S$ and $T$ such that $G[S]$ and $G[T]$ are bicliques. 
Moreover, we recall  that $G$ has a $2K_2$-partition if and only if $\overline{G}$ has a disconnected cut.
We need the following lemma, which follows from Lemma~\ref{l-4p1}.

\begin{lemma}[\cite{DMS12}]\label{l-4p1b}
{\sc $2K_2$-Partition} is $O(n^3)$-time solvable for $K_4$-free graphs.
\end{lemma}

We can now prove the following lemma.

\begin{lemma}\label{l-p1p3}
{\sc Disconnected Cut} is $O(n^3)$-time solvable for $(P_1+P_3)$-free graphs.
\end{lemma}

\begin{proof}
Let $G=(V,E)$ be a $(P_1+P_3)$-free graph on $n$ vertices. We may assume that $n\geq 4$.
We check if the $\overline{P_1+P_3}$-free graph $\overline{G}$ has a $2K_2$-partition by applying the following algorithm, which 
has running time $O(n^3)$.
If $\overline{G}$ has more than two disconnected components, $\overline{G}$ has no $2K_2$-partition.
Suppose $\overline{G}$ has exactly two components $D_1$ and $D_2$. Then each $D_i$ must have at least two vertices and must contain a complete bipartite spanning subgraph. The latter condition is true if and only if  each $\overline{D_i}$ is disconnected.
Suppose $\overline{G}$ has exactly one component. If $\overline{G}$ is $C_3$-free, then we apply Lemma~\ref{l-4p1b}. 
Otherwise, $\overline{G}$ is complete $k$-partite for some $k\geq 3$ due to Lemma~\ref{l-paw}. 
Then, as $n\geq 4$, we observe that $\overline{G}$ has a $2K_2$-partition.
\qed
\end{proof}

The graph $\overline{2P_1+P_2}$ is also known as the {\it diamond}.

\begin{lemma}\label{l-diamond}
{\sc Disconnected Cut} is $O(n^3)$-time solvable for $\overline{2P_1+P_2}$-free graphs.
\end{lemma}

\begin{proof}
Let $G$ be a $\overline{2P_1+P_2}$-free graph. We check in $O(n^2)$ time if $G$ has a dominating vertex.
If so, then $G$ has no disconnected cut due to
Lemma~\ref{l-dominating}. Assume $G$ has no dominating vertex.
As $G$ is $\overline{2P_1+P_2}$-free, the neighbourhood $N(u)$ of each vertex must be $P_3$-free, and thus $G[N(u)]$ is the disjoint union of one or more complete graphs.  If $G[N(u)]$ is the disjoint union of at least two complete graphs, then $u$ has a disconnected neighbourhood. Consequently, $G$ has a disconnected cut due to Lemma~\ref{l-disconnected}.
We check in $O(n^3)$ time if $G$ has a vertex whose neighbourhood is a disjoint union of at least two complete graphs.
Suppose $G$ has no such vertex, so every neighbourhood $N(u)$ is a clique. Then $G$ itself must be a clique and thus $G$ has no disconnected cut.
\qed
\end{proof}

We are now ready to prove the following summary for {\sc Disconnected Cut} restricted to $H$-free graphs; see also
Table~\ref{t-four}.

\begin{theorem}\label{t-k4}
Let $H\neq K_4$ be a graph on at most four vertices. Then {\sc Disconnected Cut} is polynomial-time solvable for $H$-free graphs.
\end{theorem}

\begin{proof}
We may assume that $H$ has exactly four vertices. Let $H$ have $p$ connected components. 
If $p=4$, then $H=4P_1$ and we can use Lemma~\ref{l-4p1}, proven in~\cite{DMS12}.
If $p=3$, then $H=2P_1+P_2$. This case was also proven in~\cite{DMS12}.
If $p=2$, then $H=2P_2$, or $H=P_1+P_3$, or $H=C_3+P_1$. The first case was proven in~\cite{CDEFFK10}. The second case is proved in Lemma~\ref{l-p1p3}. The third case was proven in~\cite{DMS12}.
If $p=1$, then $H=K_{1,3}$, or $H=P_4$, or $H=\overline{P_1+P_3}$, or $H=\overline{2P_1+P_2}$.
The first case follows from Theorem~\ref{thm:claw}. 
The second case follows from a result of~\cite{FMPS09}, which states that {\sc Disconnected Cut} is polynomial-time solvable for 
graphs $G=(V,E)$ with a dominating edge, that is, an edge $e=xy$ with $N_G(x)\cup N_G(y)=V$.
Every $P_4$-free graph is a cograph (and vice versa). It follows rom the definition of a connected cograph that
every connected $P_4$-free graph on at least two vertices has a spanning complete bipartite subgraph, and thus a dominating edge
(take an an edge with endpoints in each of the two partition classes of the spanning complete bipartite subgraph).
The third case is prove in Lemma~\ref{l-paw2}. The fourth case is prove in Lemma~\ref{l-diamond}.
\qed
\end{proof}

\begin{table}[h]
\begin{center}
\begin{tabular}{|l|l|}
\hline
$H$ &{Complexity of {\sc Disconnected Cut}} \\
\hline
$4P_1$ & polynomial~\cite{DMS12} \\
\hline
$2P_1+P_2$ (co-diamond) & polynomial~\cite{DMS12} \\
\hline
$2P_2$ & polynomial~\cite{CDEFFK10} \\
\hline
$P_1+P_3$ (co-paw) & polynomial  ({\bf this work}: Lemma~\ref{l-p1p3})\\
\hline
$C_3+P_1$ (co-claw) & polynomial~\cite{DMS12} \\
\hline
$K_{1,3}$ (claw) & polynomial ({\bf this work}: Theorem~\ref{thm:claw}) \\
\hline
$P_4$ & polynomial~\cite{FMPS09}  \\
\hline
$\overline{P_1+P_3}$ (paw) & polynomial  ({\bf this work}: Lemma~\ref{l-paw2}) \\
\hline
$\overline{2P_1+P_2}$ (diamond) &	polynomial  ({\bf this work}: Lemma~\ref{l-diamond}) \\
\hline
$K_4$ & \emph{open}\\
\hline
\end{tabular}
\end{center}
\medskip 
\caption{The complexity of {\sc Disconnected Cut} on $H$-free graphs when $H$ has four vertices.}
\label{t-four}
\end{table}

\section{Open Problems}\label{s-con}

In light of Corollaries~\ref{c-con} and~\ref{c-com} 
we ask about the complexities of {\sc $C_4$-Contractibility} and
{\sc ${\cal C}_4$-Compaction} for claw-free graphs of diameter at least~3.
We note that the \NP-complete problem $P_4$-{\sc Contractibility}~\cite{BV87} is polynomial-time solvable for claw-free graphs~\cite{FKP13}.

Both the complexity classification of  {\sc $H$-Compaction} and {\sc Surjective $H$-Colouring} are wide open.
In particular, it is not known if there exists a graph~$H$ for which these two problems have a different complexity. However, if we impose restrictions on the input graph, such a graph~$H$ is known: {\sc ${\cal C}_4$-Compaction} is \NP-complete for graphs of diameter~3~\cite{Vi02}, whereas {\sc Surjective ${\cal C}_4$-Colouring} (being equivalent to {\sc Disconnected Cut}) is trivial on this graph class. 
In contrast to claw-free graphs, graphs of diameter~3 do not form a hereditary graph class, that is, they are not closed under vertex deletion. This leads to the natural question if there exist a hereditary graph class ${\cal G}$ and a graph~$H$, such that {\sc $H$-Compaction} and {\sc Surjective $H$-Colouring} have different complexity when restricted to ${\cal G}$.
Should {\sc ${\cal C}_4$-Compaction} turn out to be \NP-complete for claw-free graphs, then due Theorem~\ref{thm:claw} and the equivalency between {\sc Disconnected Cut} and {\sc Surjective $H$-Colouring} we can take the class of claw-free graphs as ${\cal G}$ and
the graph~${\cal C}_4$ as $H$ to find such a pair $({\cal G},H)$.

In light of Theorem~\ref{t-k4}, we also ask what the complexity of {\sc Disconnected Cut} is for $K_4$-free graphs.

\end{document}